\documentclass[12pt]{amsart}

\input{ams_header.sty}

\usepackage{tikz-cd}
\usepackage{braket}
\usepackage{arydshln}
\usepackage{chngcntr}
\usepackage{apptools}
\AtAppendix{\counterwithin{theorem}{section}}

\DeclareMathOperator{\vspan}{span}
\newcommand{\msA}{\mathscr{A}}
\newcommand{\msB}{\mathscr{B}}
\newcommand{\POVM}{\mathscr{A}_{\textit{POVM}}} 
\newcommand{\PVM}{\mathscr{A}_{\textit{PVM}}} 
\newcommand{\wtd}{\widetilde}
\newcommand{\tr}{\text{tr}}

\newcommand{\model}{S = \big(H_A, H_B, \{M^x_a : a \in A, x \in X\}, \{N^y_b:b \in B, y \in Y\}, \ket{\psi}\big)}
\newcommand{\wtdmodel}{\wtd{S}= \big(\wtd{H}_A, \wtd{H}_B, \{\wtd{M}^x_a : a \in A, x \in X\}, \{\wtd{N}^y_b:b \in B, y \in Y\}, \ket{\wtd{\psi}}\big)}

\newcommand{\qcmodel}{S = \big(H,\{M^x_a:a\in A,x\in X \},\{N^y_b:b\in B,y\in Y\},\ket{\psi} \big)}
\newcommand{\wtdqcmodel}{\wtd{S} = \big(\wtd{H}, \{\wtd{M}^x_a : a \in A, x \in X\}, \{\wtd{N}^y_b:b \in B, y \in Y\}, \ket{\wtd{\psi}}\big)}

\title{An operator-algebraic formulation of self-testing}

\author[C. Paddock]{Connor Paddock$^{1,2}$}
\author[W. Slofstra]{William Slofstra$^{1,3}$}
\author[Y. Zhao]{Yuming Zhao$^{1,3}$}
\author[Y. Zhou]{Yangchen Zhou$^{1}$}

\address[1]{Institute for Quantum Computing, University of Waterloo}
\address[2]{Department of Combinatorics \& Optimization, University of Waterloo}
\address[3]{Department of Pure Mathematics, University of Waterloo}

\email{cpaulpad@uwaterloo.ca}
\email{weslofst@uwaterloo.ca}
\email{y658zhao@uwaterloo.ca}
\email{yangchen.zhou@uwaterloo.ca}

\begin{document}
\maketitle
\begin{abstract}
We give a new definition of self-testing for correlations in terms of
states on $C^*$-algebras. We show that this definition is equivalent to the
standard definition for any class of finite-dimensional quantum models which is
closed, provided that the correlation is
extremal and has a full-rank model in the class. This last condition
automatically holds for the class of POVM quantum models, but does not
necessarily hold for the class of projective models by a result of
Baptista, Chen, Kaniewski, Lolck, Man{\v{c}}inska, Gabelgaard Nielsen, and Schmidt. For extremal binary correlations and for extremal
synchronous correlations, we show that any self-test for projective models is a
self-test for POVM models. The question of whether there is a self-test for
projective models which is not a self-test for POVM models remains open. 

An advantage of our new definition is that it extends naturally to commuting
operator models. We show that an extremal correlation is a self-test for
finite-dimensional quantum models if and only if it is a self-test for
finite-dimensional commuting operator models, and also observe that many known
finite-dimensional self-tests are in fact self-tests for infinite-dimensional
commuting operator models. 
\end{abstract}

\section{Introduction}

In a bipartite Bell scenario, two spatially-separated parties (Alice and Bob) are given measurement settings $x$ and $y$ drawn from some finite sets $X$ and $Y$ respectively, and return measurement outcomes $a$ and $b$ drawn from finite sets $A$ and $B$. In quantum mechanics, Alice and Bob's actions in such a scenario can be modelled by measurements $\{M^x_a : a \in A\}, x \in X$ and $\{N^y_b : b \in B\}, y \in Y$ on respective local Hilbert spaces $H_A$ and $H_B$. If the joint system is in state $\ket{\psi} \in H_A \otimes H_B$, then the probability that Alice and Bob measure outcomes $a,b$ on inputs $x,y$ is 
\begin{equation*}
    p(a,b | x, y) = \braket{\psi | M^x_a \otimes N^y_b|\psi}.
\end{equation*}
The collection of probabilities $p = \{p(a,b|x,y)\}$ is called a quantum correlation, and the collection $S = (H_A,H_B,\{M^x_a\},\{N^y_b\},\ket{\psi})$ is a model for $p$. While the correlation $p$ can be directly observed from Alice and Bob's actions, the model $S$ cannot, and in fact there are typically many different models for a given correlation $p$. It is a remarkable fact that some correlations have a unique quantum model, in the sense that there is an ideal model $\wtdmodel$ for $p$, such that for any other model $\model$, there are Hilbert spaces $H_A^{aux}$, $H_B^{aux}$, isometries $I_A : H_A \to \wtd{H}_A \otimes H_A^{aux}$ and $I_B : H_B \to \wtd{H}_B \otimes H_B^{aux}$, and a state $\ket{aux} \in H_A^{aux} \otimes H_B^{aux}$ such that
\begin{equation*}
    (I_A \otimes I_B) \cdot (M^x_a \otimes N^y_b) \ket{\psi} = \left(\wtd{M}^x_a \otimes \wtd{N}^y_b \ket{\wtd{\psi}}\right) \otimes \ket{aux} 
\end{equation*}
for all $(x,y,a,b) \in X \times Y \times A \times B$. A correlation satisfying this condition is called a \emph{self-test}. The isometries $I_A$ and $I_B$ are necessary, because tensoring a model with an auxiliary register and state does not change any observable consequences of the model, including the correlation $p$. 

The definition of a self-test given above is somewhat ad-hoc: it's clear that
some type of equivalence between models is required for the definition, but why
exactly this equivalence, over this set of models? Indeed, many definitions
of self-testing have appeared since the inception of self-testing in
\cite{MY03}, with a rough consensus seeming to form around the above definition
only recently. Despite this ad-hoc nature, the definition above and its
variants have been very successful. Among other achievements, self-tests have
been used in proofs of device-independent cryptography \cite{BSCA18a,BSCA18b};
it is possible to self-test any pure bipartite state \cite{CGS17}; and
self-testing is one of the keys to the recent proof of
$\text{MIP}^*=\text{RE}$, which has resolved the Connes' embedding problem, one
of the biggest problems in operator algebras \cite{JNVWY20a}. A survey of the
field of self-testing can be found in \cite{SB20}. 

Nonetheless, the ad-hoc nature of the above definition does ask for more
explanation. Christandl, Houghton-Larsen, and Man{\v{c}}inska \cite{MH-LC22} have made
substantial progress on this front by giving an operational
definition\footnote{In this case, an operational definition is one which can be
phrased entirely in terms of physically measurable properties.} of self-testing
which, while complicated, is equivalent to the above definition. In this paper,
we take a different approach: we propose a succinct, mathematically motivated
definition of self-testing. Specifically, we say that a correlation $p$ is an
\emph{abstract state self-test} if for every $k, \ell \geq 1$, $x_1,\ldots,x_k
\in X$, $a_1,\ldots,a_k \in A$, $y_1,\ldots,y_\ell \in Y$, $b_1,\ldots,b_{\ell}
\in B$, the value
\begin{equation*}
    \braket{\psi| M^{x_1}_{a_1} \cdots M^{x_k}_{a_k} \otimes 
        N^{y_1}_{b_1} \cdots N^{y_\ell}_{b_\ell} |\psi}
\end{equation*}
is the same across all models for $p$. The name comes from a convenient
way of rephrasing this definition using abstract states on the 
tensor product $\POVM^{X,A} \otimes_{min} \POVM^{Y,B}$,
where $\POVM^{X,A}$ is the universal $C^*$-algebra generated by positive
elements $e^x_a$, $a \in A$, $x \in X$, subject to the relations $\sum_{a \in
A} e^x_a = \Id$. Specifically, a correlation $p$ belongs to the set $C_q$ of finite-dimensional quantum correlations if and only if there is a finite-dimensional
abstract state $f$ on $\POVM^{X,A} \otimes_{min} \POVM^{Y,B}$, such that $p(a,b|x,y) =
f(e^x_a \otimes e^{y}_b)$ for all $x,y,a,b$. Here an \textit{abstract state} 
(as opposed to a vector state) refers to a linear functional $f$ from the
algebra to $\C$, such that $f(a) \geq 0$ for all positive operators $a$, and
$f(\Id)=1$. An abstract state is said to be \textit{finite-dimensional} if its GNS representation
(see the next section) is finite-dimensional. With this terminology, a
correlation is an abstract state self-test if and only if there is a unique
finite-dimensional state $f$ realizing $p$.

Our first main theorem (\Cref{thm:mainresult1}) is that, if $p$ is an extreme point
of the set $C_q$, then $p$ is a self-test for finite-dimensional POVM models if
and only if it is an abstract state self-test.  Although this theorem is stated
for POVM models, the definition of self-testing is often restricted to models
with projective measurements, or other subclasses of all quantum models. This
type of restriction on the class of models can also be made for abstract states. For instance,
projective models correspond to states on the tensor product $\PVM^{X,A}
\otimes_{min} \PVM^{Y,B}$, where $\PVM^{X,A}$ is the quotient of $\POVM^{X,A}$
by the relations $(e^x_a)^2 = e^x_a$ for $x \in X$, $a \in A$, and so we can
say that a correlation $p$ is an abstract state self-test for projective models
if there is a unique finite-dimensional state on $\PVM^{X,A} \otimes_{min}
\PVM^{Y,B}$ realizing $p$ (see \cref{cor:mainresult2} part (b)). In the full
version of \cref{thm:mainresult1}, we show that the abstract state definition
and the standard definition of self-testing are equivalent for any sufficiently
nice class $\mcC$ of models, provided that $p$ has a full-rank model in $\mcC$.
Interestingly, a result of Baptista, Chen, Kaniewski, Lolck, Man{\v{c}}inska, Gabelgaard Nielsen, and Schmidt \cite{lifiting23} states
that there are correlations which do not have a full-rank projective model
(the existence of such models had been assumed in some prior work on
self-testing). Formulating our results in this generality raises some
interesting open questions, such as whether there are self-tests for the class
of projective models which are not self-tests for the class of POVM models.
While this is an open question, our second main result (\cref{thm:mainresult3})
is that in the case of binary correlations (in the sense that $|A|=|B|=2$), or
synchronous correlations (in the sense of \cite{PSSTW16}), a self-test for
projective models is a self-test for POVM models, provided $p$ is an extreme
point of the set $C_q$. This shows that many common examples of self-tests,
such as the CHSH game and the Mermin-Peres magic square, are self-tests for
POVM models. 

One of the advantages of the abstract state definition of self-testing is that
it generalizes easily to other sets of correlations. For instance, the set of
finite-dimensional quantum correlations $C_q$ is not closed \cite{Slof19b}, and
there are now several elementary proofs of this fact using self-testing
techniques \cite{Col20,Bei21}. The closure of $C_q$ is denoted by $C_{qa}$, and
a correlation $p$ belongs to $C_{qa}$ if and only if there is a (not
necessarily finite-dimensional) state $f$ on $\POVM^{X,A} \otimes_{min}
\POVM^{Y,B}$ such that $f(e^x_a \otimes e^y_b) = p(a,b|x,y)$ for all
$a,b,x,y$. Correlations $p \in C_{qa}$ do not necessarily have tensor product
models, and thus it is not clear how to generalize the standard notion of
self-testing with local isometries to this setting. With the abstract state
definition, we can say that a correlation $p \in C_{qa}$ is a self-test if
there is a unique abstract state on $\POVM^{X,A} \otimes_{min} \POVM^{Y,B}$
realizing $p$.

While not all correlations in $C_{qa}$ have tensor product models, correlations
in $C_{qa}$ can be modelled using another framework for bipartite systems,
commuting operator models. In a commuting operator model, Alice and
Bob make measurements on a joint Hilbert space (with no tensor product
structure), and their measurement operators are required to commute (in place
of factoring as tensor products). The set of correlations with a commuting
operator model is denoted by $C_{qc}$. If $p$ has a finite-dimensional
commuting operator model, then $p$ also has a finite-dimensional tensor product
model. However, the $\text{MIP}^*=\text{RE}$ result of Ji, Natarajan, Vidick,
Wright, and Yuen shows that there are correlations $p \in C_{qc}$ which are not
contained in $C_{qa}$ \cite{JNVWY20a}. There is another tensor product for
$C^*$-algebras which is associated with the commuting operator framework, the
max-tensor product, and a correlation $p$ belongs to $C_{qc}$ if and only if
there is a state $f$  on $\POVM^{X,A} \otimes_{max} \POVM^{X,A}$ with $f(e^x_a
\otimes e^y_b) = p(a,b|x,y)$ for all $a,b,x,y$. The definition of an abstract
state self-test immediately generalizes to commuting operator models.  Having a
definition of a self-test for commuting operator models raises the prospect of
being able to construct self-tests in $C_{qc}$ which have no finite-dimensional
models, perhaps from group algebras which are known to have a unique tracial
state (see, e.g., \cite{DM12}). In this work, we show that there are already
several examples of commuting operator self-tests with finite-dimensional
models. In particular, using results from \cite{Tsi87} we show that there are
many XOR correlations which are commuting operator self-tests. Notably,
this implies that the CHSH game is a commuting operator self-test. This has
independently been observed by Frei \cite{Frei22} as well.

Although one of the goals of this paper is to address the many variations of
self-testing definitions, it is still helpful to make some simplifications, and
for this reason we look only at self-testing with pure states. A recent result
of Baptista, Chen, Kaniewski, Lolck, Man{\v{c}}inska, Gabelgaard Nielsen, and Schmidt \cite{lifiting23} shows that for
extremal correlations, self-testing with pure states is equivalent to
self-testing with mixed states (using the definition from \cite{SB20}). We also
do not address robust self-testing --- what happens for models that achieve a
correlation close to $p$. While robustness is crucial in many applications of
self-testing, ignoring it at this stage allows us to concentrate on the key
points of the theory. Thus we leave robustness for later work. 

The rest of this paper is organized as follows. In \cref{sec:prelims}, we
review some background concepts and terminology. In \cref{sec:correlations}, we
state our main results for finite-dimensional models. These results are proven
in \cref{section:mainresult1} and \cref{sec:syn}. In
\cref{sec:tensor_product_models}, we show that the definition of abstract state
self-testing extends to the class of tensor product models for the set of
(quantum spatial) correlations $C_{qs}$.
In \cref{sec:commuting_operator_models} we discuss a definition of abstract
state self-testing for commuting operator models, and provide several examples
of commuting operator self-tests.

\subsection{Acknowledgements}
The authors thank Adam Bene Watts, Ranyiliu Chen, Alexander Frei, Matt Kennedy,
David Lolck, Ben Lovitz, Laura Man\v{c}inska, Vern Paulsen, and Simon Schmidt
for helpful conversations. This work was supported by NSERC DG 2018-03968 and
an Alfred P. Sloan Research Fellowship.

\section{Preliminaries}\label{sec:prelims}

We use bra-ket notation for vectors in Hilbert spaces. A \textbf{vector state}
is a unit vector in some Hilbert space, and a Hilbert space is
\textbf{bipartite} if it is a tensor product $H_A \otimes H_B$ for some Hilbert
spaces $H_A, H_B$. A Hilbert space is \textbf{separable} if it is
finite-dimensional or isomorphic to $\ell^2 (\N)$, the Hilbert space of complex
square-summable sequences.  A bipartite vector state $\ket{\psi} \in H_A
\otimes H_B$ has a \textbf{Schmidt decomposition} $\ket{\psi}=\sum_{i \in \mcI}
\lambda_i\ket{\alpha_i}\otimes \ket{\beta_i}$, where the \textbf{Schmidt
coefficients} $\lambda_i$ are (strictly) positive and have the property that
$\sum_{i \in \mcI} |\lambda_i|^2=1$, and $\{\ket{\alpha_i}\}_{i \in \mcI}$ and
$\{\ket{\beta_i}\}_{i \in \mcI}$ are orthonormal subsets of $H_A$ and $H_B$
respectively with a common index set $\mcI$. The index set $\mcI$ is either
finite or countable. The cardinality of $\mcI$ is called the \textbf{Schmidt
rank} of $\ket{\psi}$. For a derivation of the Schmidt decomposition in
infinite-dimensional spaces, see e.g. \cite[Theorem A.5]{CLP17}. 

A \textbf{positive operator-valued measure} (or \textbf{POVM}) on a Hilbert space $H$ 
with finite index set $I$ is a collection of positive operators $\{M_i:i\in I\}$ on $H$
such that $\sum_{i\in I}M_i=\Id$. We say that a POVM $\{M_i\}_{i \in I}$ is a
\textbf{projection-valued measure} (or \textbf{PVM}) if in addition the operators $M_i$
are self-adjoint projections for all $i \in I$. Naimark's dilation theorem
states that for any POVM $\{M_i\}_{i \in I}$ on a Hilbert space $H$, there is a
Hilbert space $K$, an isometry $V : H \to K$, and a PVM $\{P_i\}_{i \in I}$
such that $M_i=V^*P_i V$. One common proof of the theorem is to take
$K = H \otimes \C^{I}$, $P_i = \Id_H \otimes \ket{i}\bra{i}$, and 
$V : H \to K : \ket{h} \mapsto \sum_{i \in I} M_i^{1/2} \ket{h} \otimes
\ket{i}$, where $\ket{i}, i \in I$ is the standard orthonormal basis for
$\C^{I}$. 

For background on $C^*$-algebras, which are used throughout the paper,
see for instance \cite{B17}. We recall some of the key concepts for readers
who are less familiar with this theory. A complex $*$-algebra $\mcA$
is a unital algebra over $\C$ with an antilinear involution $a \mapsto a^*$, such
that $(a^*)^* = a$ and $(ab)^* = b^* a^*$. A \textbf{$C^*$-algebra} $\mcA$ is a
complex $*$-algebra with a submultiplicative Banach norm, such that the
$C^*$-identity $\|a^*a\|=\|a\|^2$ holds for all $a\in \mcA$. A
\textbf{representation} of a $C^*$-algebra $\mcA$ on a Hilbert space $H$ is a $*$-homomorphism
$\pi:\mcA\to \msB(H)$, where $\msB(H)$ is the $C^*$-algebra of \textbf{bounded
linear operators} acting on a Hilbert space $H$, with the operator norm
$\|\cdot\|_{op}$. A representation $\pi$ of $\mcA$ on $H$ is \textbf{faithful} if $\pi$
is injective. If $X \subseteq H$, we let 
$\pi(\mcA)X := \vspan \{\pi(a) \ket{h} : \ket{h} \in X, a \in \mcA\}$. A
\textbf{subrepresentation} of $\pi$ is a closed subspace $K \subseteq H$ such
that $\pi(\mcA) K = K$. A representation is \textbf{irreducible} if it contains no
proper non-zero subrepresentations.  The \textbf{commutant} of a subset $Z
\subseteq \msB(H)$ is $Z' := \{ T \in \msB(H) : T S = S T \text{ for all } S \in Z\}$. 
Schur's lemma states that $\pi$ is irreducible if and only if $\pi(\mcA)' =
\C \Id$, where $\pi(\mcA)'$ is the commutant of $\pi(\mcA)$. 
A vector $\ket{v} \in H$ is \textbf{cyclic} for a representation $\pi$ if
$(\pi(\mcA)\ket{v})^-=H$, where $(\cdot)^-$ denotes the closure with respect to
the Hilbert space norm. A \textbf{cyclic representation} of $\mcA$ is a tuple
$(\pi, H, \ket{v})$, where $\pi$ is a representation of $\mcA$ on $H$, and
$\ket{v} \in H$ is a cyclic vector for $\pi$. Two representations $\pi$ and
$\varphi$ of $\mcA$ on Hilbert spaces $H$ and $K$ respectively are
\textbf{(unitarily) equivalent}, denoted $\pi\cong\varphi$, if there is a unitary $U:H\to K$ such that
$U\pi(a)U^*=\varphi(a)$ for all $a\in \mcA$. Two cyclic representations
$(\pi,H,\ket{v})$ and $(\varphi, K, \ket{w})$ are unitarily equivalent if
there is such a unitary with $U \ket{v} = \ket{w}$. 

An element $a\in \mcA$ of a $C^*$-algebra $\mcA$ is \textbf{positive} if $a=b^*b$ for
some $b\in \mcA$, in which case we write $a\geq 0$. The set of positive
elements forms a closed cone in $\mcA$. 
The relation $b\leq a$ if $a-b\geq 0$ defines a partial order on the self-adjoint
elements of $\mcA$.
A linear functional $f : \mcA \to \C$ is \textbf{positive} if $f(a) \geq 0$ for
all positive elements $a \in \mcA$.  A \textbf{state} (or
\textbf{abstract state}) on a $C^*$-algebra is a positive linear functional $f$
such that $f(\Id) = 1$. The set of all states on a $C^*$-algebra $\mcA$ is
called the \textbf{state space} of $\mcA$. If
$\phi : \mcA \to \C$ is a positive linear functional on $\mcA$, then there is 
a cyclic representation $(\pi_{\phi}, H_{\phi}, \ket{\xi_\phi})$ of $\mcA$, 
called the \textbf{GNS representation} of $\phi$, such that $\phi(a) = 
\braket{\xi_{\phi}|\pi_{\phi}(a)|\xi_{\phi}}$ for all $a \in \mcA$. Conversely,
if $\phi$ is the state defined by $\phi(a) = \braket{\xi|\pi(a)|\xi}$ for
some cyclic representation $(\pi, H, \ket{\xi})$, then $(\pi, H, \ket{\xi})$
is unitarily equivalent to the GNS representation of $\phi$. Thus there is a
one-to-one correspondence between positive linear functionals and unitary
equivalence classes of cyclic representations, with states corresponding to
cyclic representations where the cyclic vector is a vector state.  Every
representation of $\mcA$ is a direct sum of cyclic representations.

In this work, a state $\phi$ is \textbf{finite-dimensional} if the Hilbert
space $H_\phi$ in the GNS representation $(\pi_\phi,H_\phi,\ket{\xi_\phi})$ is
finite-dimensional. If $(\pi,H)$ is a representation of $\mcA$ with $H$
finite-dimensional, then the double-commutant theorem states that there is an
isomorphism $H \iso \bigoplus_{i=1}^\ell \C^{n_i} \otimes \C^{m_i}$ for some
integers $n_i$, $m_i$, $i=1,\ldots,\ell$, such that 
\begin{equation}
    \pi(A)\cong \bigoplus_{i=1}^\ell M_{n_i}(\C)\otimes \Id_{m_i}\,,\text{ and}\quad \pi(A)'\cong \bigoplus_{i=1}^\ell \Id_{n_i}\otimes M_{m_i}(\C).
\end{equation}
In particular, $\pi(\mcA)$ and $\pi(\mcA)'$ are direct sums of matrix algebras. 

Given two $C^*$-algebras $\mcA$ and $\mcB$, their algebraic tensor product
$\mcA\otimes_{alg} \mcB$ is a $*$-algebra. However, there is more than one way
to make $\mcA\otimes_{alg}\mcB$ into a $C^*$-algebra. The first is to define a
norm $\|\cdot\|_{max}$ on $\mcA \otimes_{alg} \mcB$ by
\begin{equation*}
    \|a\|_{max} := \sup\{\|\pi(a)\| : \pi \text{ is a representation of } \mcA \otimes_{alg} \mcB\}.
\end{equation*}
Representations of $\mcA \otimes_{alg} \mcB$ on a Hilbert space $H$ correspond
to pairs of representations $\pi_A : \mcA \to \msB(H)$ and $\pi_B : \mcB \to
\msB(H)$, such that $\pi_A(a) \pi_B(b) = \pi_B(b) \pi_A(a)$ for all $a \in
\mcA$, $b \in \mcB$, so this gives us another way to understand $\|\cdot\|_{max}$.
After completion, this gives a $C^*$-algebra $\mcA \otimes_{max} \mcB$, called
the \textbf{max tensor product}. If instead of looking at all representations
of $\mcA \otimes_{alg} \mcB$, we look at representations of the form $\pi_A
\otimes \pi_B$ where $\pi_A$ and $\pi_B$ are representations of $\mcA$ and
$\mcB$ respectively, then we get a $C^*$-algebra $\mcA \otimes_{min} \mcB$
called the \textbf{min tensor product}. 

\section{Correlations and self-testing}\label{sec:correlations}

Let $X$, $Y$, $A$, and $B$ be finite sets. A \textbf{bipartite correlation} (or
\textbf{behaviour}) is a function $p \in \R_{\geq 0}^{A \times B \times X \times Y}$ such that 
\begin{equation*}
    \sum_{(a,b) \in A \times B} p(a,b|x,y) = 1
\end{equation*}
for all $x \in X$, $y \in Y$. A correlation can be thought of as specifying the probability of getting outputs $(a,b) \in A \times B$ from inputs $(x,y) \in X \times Y$. A \textbf{(POVM) quantum model} for a correlation $p$ is a tuple 
\begin{equation*}
    S=(H_A, H_B, \{M^x_a : a \in A, x \in X\}, \{N^y_b:b \in B, y \in Y\}, \ket{\psi}),
\end{equation*}
where
\begin{enumerate}[(i)]
    \item $H_A$ and $H_B$ are finite-dimensional Hilbert spaces, 
    \item $\{M^x_a : a \in A\}$ is a POVM on $H_A$ (meaning a collection of positive operators such that $\sum_{a \in A} M^x_a = \Id$) for all $x \in X$,
    \item $\{N^y_b : b \in B \}$ is a POVM on $H_B$ for all $y \in Y$, and
    \item $\ket{\psi} \in H_A \otimes H_B$ is a vector state (i.e. a unit vector), 
\end{enumerate}
such that
\begin{equation*}
    p(a,b|x,y) = \braket{\psi| M^x_a \otimes N^y_b|\psi}
\end{equation*}
for all $(a,b,x,y) \in A \times B \times X \times Y$.

Any tuple satisfying (i)-(iv) determines a correlation $p$ via this last
equation, and we refer to this as the correlation of the quantum model.  A
quantum model $S$ for a correlation $p$ is \textbf{projective} (or a
\textbf{PVM quantum model}) if the operators $M^x_a$ and $N^y_b$ are
self-adjoint projections for all $x,y,a,b$, or in other words if the
measurements $\{M^x_a :  a \in A\}$ and $\{N^y_b : b \in B\}$ are projective. A
quantum model is \textbf{full-rank} if $\dim H_A = \dim H_B$, and the Schmidt
rank of $\ket{\psi}$ is $\dim H_A$. The set of correlations in $\R_{\geq 0}^{A
\times B \times X \times Y}$ which have a quantum model is convex, and is usually
denoted by $C_{q}(X,Y,A,B)$, or just $C_q$ when the sets $A,B,X,Y$ are clear.
By Naimark dilation, every element of $C_q$ has a projective finite-dimensional
model, and by restricting to the support projection, every element of $C_q$
also has a full-rank (but not necessarily projective) quantum model (see \Cref{SS:selftestabs} for the definition of support projection). The sets
$C_q$ are not closed in general \cite{Slof19b}, and the closure of $C_q$ in
$\R_{\geq 0}^{A \times B \times X \times Y}$ is denoted by $C_{qa} =
C_{qa}(X,Y,A,B)$. It also makes sense to talk about models in which the Hilbert
spaces $H_A$ and $H_B$ are infinite-dimensional, and this leads to another set
of correlations $C_{qs} = C_{qs}(X,Y,A,B)$ in between $C_q$ and $C_{qa}$. We
discuss this more general class of model (which we call \textbf{tensor product
models}) in Section \ref{sec:tensor_product_models}.

The standard definition of self-testing centres around the following notion of a local dilation of a model:
\begin{definition}\label{def:localdilation}
A quantum model
\begin{align*}
        \wtd{S} = \left(\wtd{H}_A, \wtd{H}_B, \{\wtd{M}^x_a : a \in A, x \in X\}, \{\wtd{N}^y_b:b \in B, y \in Y\},\ket{\wtd{\psi}}\right)
\end{align*}
is a \textbf{local dilation} of another model
    \begin{align*}
        S = \left(H_A, H_B, \{M^x_a : a \in A, x \in X\}, \{N^y_b:b \in B, y \in Y\}, \ket{\psi}\right)
    \end{align*}
if there are finite-dimensional Hilbert spaces $H_A^{aux}$ and $H_B^{aux}$, a vector state $\ket{aux} \in H_A^{aux} \otimes H_B^{aux}$, and isometries $I_A : H_A \to \wtd{H}_A \otimes H_A^{aux}$ and $I_B : H_B \to \wtd{H}_B \otimes H_B^{aux}$ such that
    \begin{equation*}
        (I_A \otimes I_B) \cdot (M^x_a \otimes N^y_b) \ket{\psi} = \left(\wtd{M}^x_a \otimes \wtd{N}^y_b \ket{\wtd{\psi}}\right)\otimes \ket{aux}
    \end{equation*}
    for all $(a,b,x,y) \in A \times B \times X \times Y$.
    We write $S\succeq \wtd{S}$ to mean that $\wtd{S}$ is a local dilation of $S$.
\end{definition}
Note that if $S\succeq \wtd{S}$, then $S$ and $\wtd{S}$ give the same correlation. The relation $\succeq$ is a preorder, in the sense that it is transitive and reflexive, but is not anti-symmetric. We now recall the standard definition of self-testing for quantum models.
\begin{definition}\label{def:tensorproductselftest}
    Let $\mcC$ be a class of quantum models. 
    A correlation $p$ is a \textbf{self-test for the class $\mcC$} if there is a model $\wtd{S}$ for $p$ in $\mcC$, such that $S \succeq \wtd{S}$ for all other models $S$ for $p$ in $\mcC$.
\end{definition}
In \cref{def:tensorproductselftest}, we are primarily interested in the classes of projective quantum models and all quantum models, although we discuss several other classes in \cref{section:mainresult1}. Stating \cref{def:tensorproductselftest} for an arbitrary class $\mcC$ of quantum models lets us handle all these cases with one definition.
For convenience, we refer to the models $\wtd{S}$ and $S$ appearing in Definition \ref{def:tensorproductselftest} as the \textbf{ideal model} and \textbf{employed model} respectively.
In this language, a correlation is a self-test if any employed model for the correlation is equivalent to the ideal model up to local dilation.

Quantum models for correlations can also be expressed as states on
$C^*$-algebras. Given finite sets $X$ and $A$, define the \textbf{POVM algebra}
$\POVM^{X,A}$ to be the universal $C^*$-algebra generated by positive
contractions $e^x_a$, $x \in X$, $a \in A$, subject to the relations $\sum_{a
\in A} e^x_a = 1$ for all $x \in X$. If $\phi:\POVM^{X,A} \to \msB(H)$ is a
representation of $\POVM^{X,A}$ on a Hilbert space $H$, then $\{\phi(e^x_a) : a
\in A\}$ is a POVM on $H$. Conversely, given any collection $\{M^x_a : a \in
A\}$, $x \in X$ of POVMs on a Hilbert space $H$, there is a unique
representation $\phi : \POVM^{X,A} \to \msB(H)$ with $\phi(e^x_a) = M^x_a$.
To distinguish the generators of $\POVM^{X,A}$ from the generators of $\POVM^{Y,B}$
for a given tuple $(X,Y,A,B)$, we let $m^{x}_a := e^x_a \in \POVM^{X,A}$ and
$n^{y}_b := e^{y}_b \in \POVM^{Y,B}$. 
If
\begin{align*}
    S= (H_A, H_B, \{M^x_a : a \in A, x \in X\}, \{N^y_b:b \in B, y \in Y\}, \ket{\psi})
\end{align*}
is a quantum model for $p \in C_{q}$, then there is a (unique) representation $\phi_A\otimes \phi_B$ of the $C^*$-algebra $\POVM^{X,A} \otimes_{min} \POVM^{Y,B}$ with $\phi_A(m^x_a) = M^x_a$ and $\phi_B(n^y_b) = N^y_b$ for all $(a,b,x,y) \in A \times B \times X \times Y$. We call $\phi_A\otimes\phi_B$ the \textbf{associated representation of $S$}.
The abstract state $f_S$ on $\POVM^{X,A} \otimes_{min} \POVM^{Y,B}$ defined by $f_S(x): = \braket{\psi|(\phi_A \otimes \phi_B)(x)|\psi}$ is finite-dimensional, and satisfies 
\begin{equation}\label{eqn:state_to_correlation}
     f_S(m^x_a \otimes n^y_b)=\braket{\psi|\phi_A(m_a^x)\otimes \phi_B(n_b^y)|\psi}=p(a,b|x,y).
\end{equation}
We refer to $f_S$ as the \textbf{abstract state defined by $S$}. Conversely, if $f$ is a finite-dimensional state on $\POVM^{X,A} \otimes_{min} \POVM^{Y,B}$, then applying the double commutant theorem to a GNS representation of $f$ yields a quantum model $S$ such that $f = f_S$ \cite{SW08}. In particular, a correlation $p\in \R_{\geq 0}^{A \times B \times X \times Y}$ belongs to $C_{q}$ if and only if there is a finite-dimensional state $f$ on $\POVM^{X,A} \otimes_{min} \POVM^{Y,B}$ with $f(m_a^x \otimes n^y_b) = p(a,b|x,y)$ for all $(a,b,x,y) \in A \times B \times X \times Y$. Consequently, a correlation $p$ belongs to $C_{qa}$ if and only if there is a (not necessarily finite-dimensional) state $f$ on $\POVM^{X,A} \otimes_{min} \POVM^{Y,B}$ with $f(m_a^x \otimes n^y_b) = p(a,b|x,y)$ for all $(a,b,x,y) \in A \times B \times X \times Y$ \cite{Fri12,JNPPSW11}.

Other classes of models can be identified with subsets of states on $\POVM^{X,A} \otimes_{min} \POVM^{Y,B}$. For instance, let $\PVM^{X,A}$ be the quotient of $\POVM^{X,A}$ by the relations $(e^x_a)^2 = e^x_a$ for all $x \in X$, $a \in A$ (the resulting algebra $\PVM^{X,A}$ is isomorphic to the group $C^*$-algebra $C^* (\Z_m^{*n})$, where $m = |A|$ and $n = |X|$). If
\begin{equation*}
    q : \POVM^{X,A} \otimes_{min} \POVM^{Y,B} \to \PVM^{X,A} \otimes_{min} \PVM^{Y,B}
\end{equation*}
is the quotient homomorphism, and $f$ is a state on $\PVM^{X,A} \otimes_{min} \PVM^{Y,B}$, then $f \circ q$ is a state on $\POVM^{X,A} \otimes_{min} \POVM^{Y,B}$. The pullback map $q^* : f \mapsto f \circ q$ is an injection, and hence identifies states on $\PVM^{X,A} \otimes_{min} \PVM^{Y,B}$ with a subset of states on $\POVM^{X,A} \otimes_{min} \POVM^{Y,B}$. We say that a state on $\POVM^{X,A} \otimes_{min} \POVM^{Y,B}$ is \textbf{projective} if it belongs to the image of $q^*$. A finite-dimensional state $f$ on $\POVM^{X,A} \otimes_{min} \POVM^{Y,B}$ is projective if and only if $f = f_S$ for some projective quantum model $S$.

The first main contribution of this paper is a definition of self-testing in terms of states on $\POVM^{X,A} \otimes_{min} \POVM^{Y,B}$. 
\begin{definition}\label{def:abstract_self-test}
    Let $\mcS$ be a subset of states on $\POVM^{X,A} \otimes_{min} \POVM^{Y,B}$. A correlation $p$ is an \textbf{abstract state self-test for $\mcS$} if there exists a unique abstract state $f \in \mcS$ with correlation $p$.
\end{definition}

Our first main theorem is that, in many cases, \cref{def:abstract_self-test} is equivalent to the standard definition of self-testing in \cref{def:tensorproductselftest}. To state this theorem in as much generality as possible, we make the following definitions:
\begin{definition}\label{def:submodel1}
Let 
\begin{align*}
    &\model \text{ and}\\
    &\wtdmodel
\end{align*}
be two quantum models.
\begin{enumerate}
    \item We say that $S$ and $\wtd{S}$ are \textbf{locally equivalent}, and write $S\cong \wtd{S}$, if there are local unitaries $U_A:H_A\arr\widetilde{H}_A$ and $U_B:H_B\arr\widetilde{H}_B$ such that
\begin{enumerate}[(i)]
    \item $U_A\otimes U_B\ket{\psi}=\ket{\widetilde{\psi}}$, and  

    \item $U_AM^x_aU_A^*=\wtd{M}^x_a$ and $  U_BN^y_bU_B^*= \wtd{N}^y_b$ for
        all $(a,b,x,y)\in A\times B\times X\times Y$.  
\end{enumerate}
\item We say $\wtd{S}$ is a \textbf{submodel} of $S$ if there are projections
    $\Pi_A\in\mathscr{B}(H_A)$ and $\Pi_B\in\mathscr{B}(H_B)$ such that
    $\Pi_AH_A=\wtd{H}_A,\Pi_BH_B=\wtd{H}_B$, $\Pi_A\otimes \Pi_B\ket{\psi}=\lambda
    \ket{\wtd{\psi}}$ for some $\lambda\neq 0$, $\Pi_AM^x_a\Pi_A=\wtd{M}^x_a$,
    $\Pi_B N^y_b \Pi_B=\wtd{N}^y_b$, and $[\Pi_A,M^x_a]=[\Pi_B,N^y_b]=0$ for all
    $(a,b,x,y)\in A\times B \times X \times Y$.

\end{enumerate}
A class of quantum models $\mcC$ is \textbf{closed} if for any $S\in \mcC$, every $\wtd{S}$ that is locally equivalent to a submodel of $S$ is in $\mcC$. 

\end{definition}

We can now state our first main theorem:
\begin{theorem}\label{thm:mainresult1}
Let $p$ be an extreme point in $C_q$, let $\mcC$ be a closed class of quantum models that contains a full-rank model for $p$, and let $\mcS:=\{f_S:S\in\mcC\}$ be the set of finite-dimensional states induced by $\mcC$. Then $p$ is a self-test for $\mcC$ if and only if $p$ is an abstract state self-test for $\mcS$.
\end{theorem}
The proof of \Cref{thm:mainresult1} follows from \Cref{prop:onlyifdirection} and \Cref{thm:uniquestate}.
The hypotheses of Theorem \ref{thm:mainresult1} can actually be weakened a bit
further: the ``if'' direction holds as long as $p$ is extreme and $\mcC$ is closed, even if $\mcC$ does not contain a full-rank model for $p$, while the
``only if'' direction only requires that $\mcC$ is closed and contains a full-rank model for $p$.
Although the ``only if'' direction does not require that $p$ be an extremal correlation, this is necessary for the ``if'' direction, as shown by
\cref{ex:extremal_cor_not_self-test}.

The class of all quantum models and the class of projective quantum models are both closed. As discussed above, if $\mcC$ is the class of all quantum models, then $\{f_S : S \in \mcC\}$ is the set of all finite-dimensional states on $\POVM^{X,A} \otimes_{min} \POVM^{Y,B}$, while if $\mcC$ is the class of projective quantum models then $\{f_S : S \in \mcC\}$ is the set of projective finite-dimensional states. Also, every correlation in $C_q$ has a full-rank quantum model. Thus for these two classes, \cref{thm:mainresult1} implies:
\begin{corollary}\label{cor:mainresult2}
Suppose $p\in C_q(X,Y,A,B)$ is an extreme point. Then:
\begin{enumerate}[(a)]
    \item The correlation $p$ is a self-test for the class of quantum models if and only if $p$ is an abstract state self-test for finite-dimensional states. \item If $p$ has a full-rank projective quantum model, then $p$ is a self-test for projective quantum models if and only if $p$ is an abstract state self-test for projective finite-dimensional states. 
\end{enumerate}
\end{corollary}
The proof of \cref{cor:mainresult2} is given in \cref{section:mainresult2}. In
part (b), if $p$ is an abstract state self-test for projective
finite-dimensional states, then $p$ is a self-test for projective quantum
models even if $p$ does not have a full-rank projective quantum model. We do
not know whether the hypothesis that $p$ have a full-rank projective quantum
model is required for the other direction. Baptista, Chen, Kaniewski, Lolck, Man{\v{c}}inska, Gabelgaard Nielsen, and Schmidt
\cite{lifiting23} have constructed a correlation which is a self-test for full-rank
(POVM) quantum models, and a self-test for PVM quantum models, but does not have a full-rank projective quantum model,
so this condition cannot automatically be assumed to hold. 

In the next section, we'll show that if a correlation $p \in C_q$ is a self-test for all quantum models (or equivalently, an abstract state self-test for the set of all finite-dimensional states on $\POVM^{X,A} \otimes_{min} \POVM^{Y,B}$), then it has an ideal projective model, and hence is a self-test for projective quantum models. We do not know whether there is a correlation $p \in C_q$ which is a self-test for projective quantum models but not a self-test for POVM quantum models (although we'll show in the next section that any such example cannot have a full-rank projective quantum model). Our second main result is that, for two large classes of correlations, every self-test for projective quantum models is also a self-test for POVM quantum models. Recall that a correlation $p\in C_q(X,Y,A,B)$ is said to be a \textbf{synchronous correlation} if $A=B,X=Y$, and $p(a,b|x,x)=0$ whenever $a\neq b$. Recall that $p\in C_q(X,Y,A,B)$ is said to a \textbf{binary correlation} if $\abs{A}=\abs{B}=2$.

\begin{theorem}\label{thm:mainresult3}
If $p$ is a synchronous or binary correlation, and is an extreme point in
$C_q$, then the following statements are equivalent:
\begin{enumerate}
    \item $p$ is a self-test for quantum models.
    \item $p$ is a self-test for projective quantum models.
    \item $p$ is an abstract state self-test for finite-dimensional states.
    \item $p$ is an abstract state self-test for projective finite-dimensional states.
\end{enumerate}
\end{theorem}
In particular, this implies that many prominent self-tests for projective
models, such as the CHSH game, Mermin-Peres magic square game, and so on, are
also self-tests for POVM quantum models. When $p$ is a synchronous correlation,
the proof of \Cref{thm:mainresult3} shows that (1) and (2) are equivalent even
if $p$ is not an extreme point in $C_q$, as are (3) and (4) (see \Cref{prop:synchronousmain3}
and \Cref{prop:synchronousprojectivestate}).

To finish this section, we note that a major advantage of \cref{def:abstract_self-test} is that it naturally generalizes to a definition of self-testing for commuting operator models. We discuss this in  \cref{sec:commuting_operator_models}.

\section{Proof of Theorem \ref{thm:mainresult1}}\label{section:mainresult1}

In this section we prove \Cref{thm:mainresult1} and \Cref{cor:mainresult2}.
Unless stated otherwise, all the Hilbert spaces in this section are assumed to be finite-dimensional.

\subsection{Self-tests are abstract state self-tests}\label{SS:selftestabs}

We start by proving the ``only if'' direction of \Cref{thm:mainresult1}, that a
self-test (as in \Cref{def:tensorproductselftest}) is an abstract state
self-test. For the purposes of this proof, we use the following terminology.
If
\begin{align*}
    S= (H_A, H_B, \{M^x_a : a \in A, x \in X\}, \{N^y_b:b \in B, y \in Y\}, \ket{\psi})
\end{align*}
is a quantum model, the \textbf{support} of $\ket{\psi}$ in $H_A$ (resp. $H_B$)
is the image of the reduced density matrix $\rho_A :=
\tr_{H_B}(\ket{\psi}\bra{\psi})$ (resp. $\rho_B =
\tr_{H_A}(\ket{\psi}\bra{\psi})$. The \textbf{support projections} $\Pi_A$ and $\Pi_B$ of $\ket{\psi}$
are the self-adjoint projections onto the support of $\ket{\psi}$ in $H_A$ and $H_B$ respectively.
Restricting to the support of $\ket{\psi}$ gives another quantum model
\begin{equation*}
    S':=(\Pi_AH_A,\Pi_BH_B,\{\Pi_AM^x_a\Pi_A\},\{\Pi_BN^y_b\Pi_B\},\ket{\psi}),
\end{equation*}
which we call the \textbf{support model} of $S$. The support model $S'$ is a full-rank model, but is not
necessarily a submodel of $S$ if $\Pi_A$ and $\Pi_B$ do not commute with all the measurement operators
$M^x_a$ and $N^y_b$. We say that $S$ is \textbf{centrally supported} if 
\begin{align*}
    [\Pi_A,M^x_a]=[\Pi_B,N^y_b]=0 
\end{align*}
for all $(a,b,x,y)\in A\times B\times X\times Y$. The support model $S'$ of a quantum model $S$ is a 
submodel of $S$ if and only if $S$ is centrally supported. Any full-rank model is centrally supported, since the
support projections are the identity operators. Conversely, if $S$ is centrally supported, then its
support model is a full-rank submodel of $S$. Hence we have:
\begin{lemma}\label{lemma:full-rank}
    Let $\mcC$ be a closed class of quantum models, and
    let $p\in C_q$. Then $\mcC$ contains a centrally supported model for $p$ if
    and only if $\mcC$ contains a full-rank model for $p$.
\end{lemma}

For later use, we note:

\begin{lemma}\label{lem:supportmodel}
    If $S$ is a centrally supported model with support model $S'$, then $S'\succeq S$ and $S \succeq S'$. 
\end{lemma}
\begin{proof}
    Suppose $\model$ is a centrally supported quantum model.  Let $\Pi_A$ and $\Pi_B$ be the support projections of $\ket{\psi}$,
    and let $S' = (H_A', H_B', \{\widehat{M}^x_a\}, \{\widehat{N}^y_b\}, \ket{\psi})$ be the
    support model for $S$, so $H_A' = \Pi_A H_A$, $H_B' = \Pi_B H_B$, $\widehat{M}^x_a = \Pi_A M^x_a \Pi_A$,
    and $\widehat{N}^y_b = \Pi_B N^y_b \Pi_B$. Since
    \begin{align*}
        \widehat{M}^x_a\otimes\widehat{N}^y_b\ket{\psi}=M^x_a\otimes N^y_b\ket{\psi} \text{ for all }a\in A,b\in B,x\in X,y\in Y,
    \end{align*}
    it follows that $S'\succeq S$.
    
    To see $S\succeq S'$, pick an arbitrary unit vector
    $\ket{w_A} \in H'_A$, and define an isometry $I_A : H_A \to H'_A \otimes
    H_A$ by $I_A(v) = v \otimes \ket{w_A}$ if $v \in H'_A$, and $I_A(v) = \ket{w_A} \otimes v$
    if $v \in (H'_A)^{\perp}$. Similarly, define an isometry $I_B : H_B \to H'_B \otimes H_B$ using
    a vector $\ket{w_B} \in H'_B$. Using the fact that $S$ is centrally supported, 
    \begin{align*}
        (I_A \otimes I_B) M^x_a \otimes N^y_b \ket{\psi} & = (I_A \otimes I_B) \Pi_A M^x_a \Pi_A \otimes \Pi_B  N^y_b \Pi_B \ket{\psi} \\
            & = (\widehat{M}^x_a \otimes \widehat{N}^y_b \ket{\psi}) \otimes \ket{w_A} \otimes \ket{w_B}
    \end{align*}
    for all $x \in X$, $a \in A$, $y \in Y$, $b \in B$, so $S \succeq S'$.
\end{proof}

We will now provide a useful characterization of centrally supported quantum models. To do this, we first need to establish the following lemmas.

\begin{lemma}\label{lemma:onlyif}
Let $\ket{\psi}\in H_A\otimes H_B$ be a bipartite vector state, and let $\Pi_A$
and $\Pi_B$ be the support projections of $\ket{\psi}$. For any self-adjoint
operator $E\in\msB(H_A)$, if there exists an operator $\widehat{E}\in\msB(H_B)$
such that $E\otimes \Id\ket{\psi}=\Id\otimes \widehat{E}\ket{\psi}$, then the
support projection $\Pi_A$ commutes with $E$.
\end{lemma}
\begin{proof}
    Let $\ket{\psi}=\sum_{i=1}^k\lambda_{i}\ket{\alpha_i}\otimes\ket{\beta_i}$ be a
    Schmidt decomposition of $\ket{\psi}$, so the support projection of $\ket{\psi}$ on $H_A$ is $\Pi_A =
    \sum_{i=1}^k\ket{\alpha_i}\bra{\alpha_i}$. Suppose there exists an
operator $\widehat{E}\in\msB(H_B)$ such that $E\otimes \Id\ket{\psi}=\Id\otimes
    \widehat{E}\ket{\psi}$. Then 

    \begin{align*}
        \Pi_AE\otimes\Id\ket{\psi}=\Pi_A\otimes\widehat{E}\ket{\psi}=\Id\otimes \widehat{E}\ket{\psi}=E\otimes \Id\ket{\psi},
    \end{align*}
    or in other words $(\Id-\Pi_A)E\otimes\Id\ket{\psi}=0$. This implies
     \begin{align*}
        \sum_{i=1}^k\lambda_i\big((\Id-\Pi_A)E\ket{\alpha_i}\big)\otimes\ket{\beta_i}=0, 
     \end{align*}
      and since $\lambda_i>0$ for all $1\leq i\leq k$ and
        $\{\ket{\beta_i}:1\leq i\leq k\}$ is an orthonormal subset, it follows 
    that $(\Id-\Pi_A)E\ket{\alpha_i}=0$ for all $1\leq i\leq k$. Hence,
      \begin{equation*}
          (\Id-\Pi_A)E\Pi_A=\sum_{i=1}^k(\Id-\Pi_A)E\ket{\alpha_i}\bra{\alpha_i}=0,
      \end{equation*} 
    and $E\Pi_A=\Pi_AE\Pi_A$. Since $E$ is self-adjoint, $\Pi_A E = (E \Pi_A)^* = \Pi_A E \Pi_A = E \Pi_A$. 
\end{proof}

\begin{lemma}\label{lemma:if}
If $\ket{\psi}\in H_A\otimes H_B$ is a full-rank vector state, then for any
$E\in\msB(H_A)$  (resp. $F\in\msB(H_B)$) there is an operator
$\widehat{E} \in \msB(H_B)$ (resp. $\widehat{F} \in \msB(H_A)$) such that
$E\otimes\Id\ket{\psi}=\Id\otimes\widehat{E}\ket{\psi}$ (resp. $\Id\otimes
F\ket{\psi}=\widehat{F}\otimes \Id\ket{\psi}$).
\end{lemma}
\begin{proof}
Since the vector state $\ket{\psi}\in H_A\otimes H_B$ is full-rank, we can assume without loss of generality
that $H_A=H_B= \C^k$ and $\ket{\psi}$ has Schmidt decomposition
\begin{align*}
    \ket{\psi}=\sum_{i=1}^k\lambda_i\ket{i}\otimes\ket{i}.
\end{align*}
Let $\lambda:=\sum_{i=1}^k\lambda_i\ket{i}\bra{i}$ and
$\ket{\tau}:=\sum_{i=1}^k\ket{i}\otimes\ket{i}$. Then $\lambda$ is invertible
and $\ket{\psi}=\Id\otimes\lambda\ket{\tau}=\lambda\otimes\Id\ket{\tau}$. 
Given $E\in\msB(H_A)$, let $\widehat{E}:=\lambda E^{T}\lambda^{-1}$, where $E^T$ is the
transpose of $E$ in the standard basis. Since
$\Id \otimes M^T \ket{\tau} = M \otimes \Id \ket{\tau}$ for any 
operator $M$, 
\begin{align*}
    \Id\otimes\widehat{E}\ket{\psi}&=(\Id\otimes\lambda E^{T}\lambda^{-1})(\Id\otimes\lambda)\ket{\tau}
    =\Id\otimes\lambda E^{T}\ket{\tau}
    =E\otimes\lambda\ket{\tau}
    =E\otimes\Id\ket{\psi}.
\end{align*}
The existence of $\widehat{F}$ for $F$ follows similarly.
\end{proof}

\begin{proposition}\label{prop:centrallysupported}
    A quantum model $$\model$$ is centrally supported if and only if for every $a\in A$, $x\in X$ (resp. $b\in B$, $y\in Y$) there exists an operator $\widehat{M}_a^x\in \msB(H_B)$ (resp. $\widehat{N}_b^y\in \msB(H_A)$) such that $M_a^x\otimes \Id|\psi\rangle=\Id\otimes \widehat{M}_a^x|\psi\rangle$ (resp. $\Id \otimes N_b^y\ket{\psi}=\widehat{N}_b^y\otimes \Id\ket{\psi}$) .
\end{proposition}
\begin{proof}
The ``if" part follows directly from \cref{lemma:onlyif}. For the ``only if"
part, let $\Pi_A$ and $\Pi_B$ be the support projections of $\ket{\psi}$, and
suppose $[\Pi_A,M^x_a]=[\Pi_B,N^y_b]=0$ for all $(a,b,x,y)\in A\times B\times
X\times Y$. Since $\Pi_A\otimes\Pi_B\ket{\psi}$ has full-rank in
$(\Pi_AH_A)\otimes(\Pi_BH_B)$, by \cref{lemma:if} there is an operator
$\widehat{M}^x_a$ in $\msB(\Pi_BH_B)$ (which is also an operator in $\msB(H_B)$
acting trivially on $(\Id-\Pi_B)H_B$) such that
$\Pi_AM^x_a\Pi_A\otimes\Pi_B\ket{\psi}=\Pi_A\otimes
\widehat{M}^x_a\Pi_B\ket{\psi}$ for any $a\in A, x\in X$. Therefore
\begin{align*}
    M^x_a\otimes\Id\ket{\psi}&=M^x_a\Pi_A\otimes\Pi_B\ket{\psi}=\Pi_AM^x_a\Pi_A\otimes\Pi_B\ket{\psi}
    =\Pi_A\otimes \widehat{M}^x_a\Pi_B\ket{\psi}=\Id\otimes\widehat{M}^x_a\ket{\psi}.
\end{align*}
Similarly, for any $b\in B,y\in Y$ there is $\widehat{N}_b^y\in \mathcal{B}(H_A)$ such that $\Id \otimes N_b^y\ket{\psi}=\widehat{N}_b^y\otimes \Id\ket{\psi}$.
\end{proof}

This characterization allows us to show that centrally supported models can only locally dilate to centrally supported models:

\begin{proposition}\label{prop:centrallysupportediff}
    Let $S$ and $\wtd{S}$ be two quantum models. If $S$ is centrally supported and $S\succeq\wtd{S}$, then $\wtd{S}$ is
    centrally supported.
\end{proposition}
\begin{proof}
Let 
\begin{align*}
    &\model \text{ and }\\
    &\wtdmodel
\end{align*}
be two quantum models for a quantum correlation $p$ such that $S$ is centrally supported and $S\succeq
\wtd{S}$ with local isometries $I_A,I_B$ and vector state $\ket{aux}\in
H_A^{aux}\otimes H_B^{aux}$. Let $\Pi_A,\Pi_B$ and $\wtd{\Pi}_A,\wtd{\Pi}_B$ be the support projections of $S$ and $\wtd{S}$ respectively. Then
\begin{align*}
    S':=\big(\wtd{H}_A\otimes H_A^{aux},\wtd{H}_B\otimes H_B^{aux},\{\wtd{M}^x_a\otimes\Id_{H_A^{aux}}\},\{\wtd{N}^y_b\otimes\Id_{H_B^{aux}}\},\ket{\wtd{\psi},aux} \big)
\end{align*}
is a quantum model for $p$ with support projections $\wtd{\Pi}_A\otimes \Pi_A^{aux}$ and $\wtd{\Pi}_B\otimes \Pi_B^{aux}$ where $\Pi_A^{aux}\in\msB(H_A^{aux})$ and $\Pi_B^{aux}\in\msB(H_B^{aux})$ are the support projections of $\ket{aux}$. Since $S$ is centrally supported, for any $x\in X,a\in A$, by \Cref{prop:centrallysupported}, there is an operator $\widehat{M}^x_a\in\msB(H_B)$ such that $M^x_a\otimes \Id\ket{\psi}=\Id\otimes \widehat{M}^x_a\ket{\psi}$. Since $I_AI_A^*\otimes I_BI_B^*\ket{\wtd{\psi},aux}=\ket{\wtd{\psi},aux}$ and $I_AI_A^*\otimes\Id\geq I_AI_A^*\otimes I_BI_B^*$, we see that $I_AI_A^*\otimes\Id\ket{\wtd{\psi},aux}=\ket{\wtd{\psi},aux}$. Then $I_B\widehat{M}^x_aI_B^*$ is an operator in $\msB(\wtd{H}_B\otimes H_B^{aux})$ such that
\begin{align*}
    \Id\otimes I_B\widehat{M}^x_aI_B^*\ket{\wtd{\psi},aux}&=I_AI_A^*\otimes I_B\widehat{M}^x_aI_B^*\ket{\wtd{\psi},aux}\\
    &= (I_A\otimes I_B)(\Id\otimes \widehat{M}^x_a\ket{\psi})\\
    &= (I_A\otimes I_B)(M^x_a\otimes \Id\ket{\psi})\\
    &=(\wtd{M}^x_a\otimes\Id_{H_A^{aux}})\otimes \Id\ket{\wtd{\psi},aux}.
\end{align*}
By \Cref{lemma:onlyif} we see that $\wtd{\Pi}_A\otimes\Pi_A^{aux}$ commutes with $\wtd{M}^x_a\otimes\Id_{H_A^{aux}}$, and hence $[\wtd{\Pi}_A,\wtd{M}^x_a]=0$ for all $x\in X,a\in A$. Similarly, $[\wtd{\Pi}_B,\wtd{N}^y_b]=0$ for all $y\in Y,b\in B$. We conclude that $\wtd{S}$ is centrally supported.
\end{proof}

\begin{corollary}\label{cor:supportmodel}
    If $p$ is a self-test for a class of models $\mcC$, and $p$ has a centrally supported model in $\mcC$, then every ideal model of $p$ is centrally supported. If in addition, $\mcC$ is closed, then $p$ has a full-rank ideal model in $\mcC$.
\end{corollary}
\begin{proof}
    Suppose $p$ has ideal model $\wtd{S}$. By \Cref{prop:centrallysupportediff}, $\wtd{S}$ must be centrally supported. If $\mcC$ is closed, then by \Cref{lemma:full-rank,lem:supportmodel}, 
    the support model $S'$ of $\wtd{S}$ is in $\mcC$ and is a full-rank ideal model for $p$.
\end{proof}
We prove stronger versions of \Cref{cor:supportmodel} for the class of all
quantum models in Propositions \ref{prop:alltoprojective} and \ref{prop:consequences}.

\begin{proposition}\label{prop:monomial}
Let 
\begin{align*}
   &\model \text{ and }\\
   &\wtdmodel
\end{align*}
be two quantum models for a correlation $p\in C_q(X,Y,A,B)$, with associated
representations $\pi_A\otimes\pi_B$ and $\wtd{\pi}_A\otimes\wtd{\pi}_B$
respectively. Suppose $S\succeq\widetilde{S}$ via local isometries $I_A$ and
$I_B$ and vector state $\ket{aux}\in H_A^{aux}\otimes H_B^{aux}$.  If 
$\wtd{S}$ is centrally supported, then
\begin{equation}\label{eq:monomial}
    (\wtd{\pi}_A(\alpha)\otimes\Id_{\wtd{H}_B}\ket{\wtd{\psi}})\otimes\ket{aux}=I_A\pi(\alpha)I_A^*\otimes\Id_{\wtd{H}_B\otimes H_B^{aux}}\ket{\wtd{\psi},aux} \text{ for all }\alpha\in\POVM^{X,A},  
\end{equation}
\begin{equation}\label{eq:monomial2}
    (\Id_{\wtd{H}_A}\otimes\wtd{\pi}_B(\beta)\ket{\wtd{\psi}})\otimes\ket{aux}=\Id_{\wtd{H}_A\otimes H_A^{aux}}\otimes I_B\pi(\beta)I_B^*\ket{\wtd{\psi},aux} \text{ for all }\beta\in\POVM^{Y,B}, 
\end{equation}
and $f_S=f_{\wtd{S}}$.
\end{proposition}

\begin{proof}
\Cref{eq:monomial} holds for all monomials $\alpha=\alpha_1\cdots
\alpha_k\in \POVM^{X,A}$, 
$\alpha_i\in\{e^x_a:x\in X,a\in A \}$ by induction on the monomial degree $k \in \N$. 
Indeed, the cases $k=0,1$ follow straight from \cref{def:localdilation}.
Suppose that \Cref{eq:monomial} holds for all monomials in $\POVM^{X,A}$ of
degree $k$. For any monomial $\alpha=\alpha_1\cdots\alpha_k\alpha_{k+1}$ in
$\POVM^{X,A}$ of degree $k+1$, let $\alpha':=\alpha_1\cdots\alpha_k$. 
Since $\wtd{S}$ is centrally supported, by \cref{prop:centrallysupported} there is an
$\wtd{F}\in\msB(\wtd{H}_B)$ such that
$\wtd{\pi}_A(\alpha_{k+1}) \otimes\Id\ket{\wtd{\psi}}=\Id\otimes\wtd{F}\ket{\wtd{\psi}}$. 
Thus by the inductive hypothesis,
\begin{align*}
   \big(\wtd{\pi}_A(\alpha)\otimes\Id_{\wtd{H}_B}\ket{\wtd{\psi}}\big)\otimes\ket{aux} 
    & = \big(\wtd{\pi}_A(\alpha')\wtd{\pi}_A(\alpha_{k+1})\otimes\Id_{\wtd{H}_B}\ket{\wtd{\psi}}\big)\otimes\ket{aux} \\
    & =\big(\wtd{\pi}_A(\alpha')\otimes \wtd{F}\ket{\wtd{\psi}}\big)\otimes\ket{aux} \\
   & = I_A\pi_A(\alpha')I_A^*\otimes \wtd{F}\otimes\Id_{H_B^{aux}}\ket{\wtd{\psi},aux} \\
   & =\big(I_A\pi_A(\alpha')I_A^*(\wtd{\pi}(\alpha_{k+1})\otimes\Id_{H_A^{aux}})\big)\otimes \Id_{\wtd{H}_B\otimes H_B^{aux}}\ket{\wtd{\psi},aux} \\
   & =\big(I_A\pi_A(\alpha')I_A^*I_A\pi_A(\alpha_{k+1})I_A^*\big)\otimes \Id_{\wtd{H}_B\otimes H_B^{aux}}\ket{\wtd{\psi},aux} \\
   & =I_A\pi_A(\alpha)I_A^*\otimes\Id_{\wtd{H}_B\otimes H_B^{aux}}\ket{\wtd{\psi},aux}.
\end{align*}
We conclude that \Cref{eq:monomial} holds for all $\alpha$. The proof of \Cref{eq:monomial2} is similar.
Hence, for any $\alpha\in \POVM^{X,A}$ and $ \beta\in\POVM^{Y,B}$, 
\begin{align*}
    f_{\wtd{S}}(\alpha\otimes\beta)&=\bra{\wtd{\psi}}\wtd{\pi}_A(\alpha)\otimes\wtd{\pi}_B(\beta)\ket{\wtd{\psi}}\\
    &=\bra{\wtd{\psi},aux}\wtd{\pi}_A(\alpha)\otimes\wtd{\pi}_B(\beta)\otimes\Id_{H_A^{aux}\otimes H_B^{aux}}\ket{\wtd{\psi},aux}\\
    &=\bra{\wtd{\psi},aux} I_A\pi_A(\alpha)I_A^*\otimes I_B\pi_B(\beta)I_B^*\ket{\wtd{\psi},aux}\\
    &=\bra{\psi}\pi_A(\alpha)\otimes\pi_B(\beta)\ket{\psi}\\
    &=f_S(\alpha\otimes\beta),
\end{align*}
which implies $f_S=f_{\wtd{S}}$.
\end{proof}

\begin{remark}\label{rmk:Scentral}
Together with \Cref{prop:centrallysupportediff}, \Cref{prop:monomial} implies that if $S$ is centrally supported and $S
    \succeq \wtd{S}$, then $f_S = f_{\wtd{S}}$. Although we won't use this fact, it is also possible to prove this by showing
    that \Cref{eq:monomial} and \Cref{eq:monomial2} hold when $S$ is centrally supported
    and $S \succeq \wtd{S}$.
\end{remark}

We can now prove the ``only if" direction of \Cref{thm:mainresult1} (with the weaker hypotheses mentioned after the theorem statement). 
\begin{proposition}\label{prop:onlyifdirection}
Let $p\in C_q(X,Y,A,B)$. Let $\mcC$ be a class of quantum models that is closed and contains a full-rank model for $p$, and let $\mcS=\{f_S:S\in\mcC\}$. If $p$ is a self-test for $\mcC$ then $p$ is an abstract state self-test for $\mcS$.
\end{proposition}
\begin{proof}
    Let $S$ be a full-rank model in $\mcC$ for $p$. Suppose $p$ is a self-test for
    $\mcC$, and let $\wtd{S}\in\mcC$ be an ideal model for $p$, so  $S\succeq
    \wtd{S}$. Full-rank models are centrally supported, so
$\wtd{S}$ is centrally supported by \Cref{prop:centrallysupportediff}. Then \Cref{prop:monomial} implies  that $f_{\wtd{S}}$ is the unique state in $\mcS$ achieving $p$.
\end{proof}

\subsection{Abstract state self-tests are self-tests}\label{sec:unique_state_implies_self-test}

Before we can establish the ``if'' direction of \cref{thm:mainresult1} we need
to review some basic facts about the representation theory of $C^*$-algebras.
For more details we refer the reader to \cite{B17}.

\begin{lemma}\label{lemma:irrep}
Let $\pi_A:\mcA\arr \msB(H_A)$ and $\wtd{\pi}_A:\mcA\arr \msB(\wtd{H}_A)$ be two representations of a $C^*$-algebra $\mcA$, and let $\pi_B:\mcB\arr \msB(H_B)$ and $\wtd{\pi}_B:\mcB\arr \msB(\wtd{H}_B)$ be two representations of a $C^*$-algebra $\mcB$.
\begin{enumerate}
    \item $\pi_A\otimes\pi_B$ is irreducible if and only if $\pi_A$ and $\pi_B$ are irreducible.

    \item Suppose $\pi_A$ and $\pi_B$ are irreducible. Let $\ket{\psi} $ and
        $\ket{\psi'}$ be two vector states in $H_A\otimes H_B$. If $(\pi_A\otimes
        \pi_B,H_A\otimes H_B,\ket{\psi})$ and $(\pi_A\otimes\pi_B,H_A\otimes
        H_B,\ket{\psi'})$ are GNS representations for the same state, then
        $\ket{\psi}\bra{\psi}=\ket{\psi'}\bra{\psi'}$ (in other words, the states
        $\ket{\psi}$ and $\ket{\psi'}$ only differ by a complex phase). 

     \item Suppose $\pi_A,\wtd{\pi}_A,\pi_B$ and $\wtd{\pi}_B$ are irreducible. Then
        $\pi_A\otimes\pi_B\cong \wtd{\pi}_A\otimes\wtd{\pi}_B$ if and only if
        $\pi_A\cong\wtd{\pi}_A$ and $\pi_B\cong\wtd{\pi}_B$.

\end{enumerate}
\end{lemma}
\begin{proof}
    Part (1) is an immediate consequence of Schur's lemma, along with the fact
    that $\big(\pi_A(\mcA)\otimes\pi_B(\mcB)\big)'=\big(\pi_A(\mcA)\big)'\otimes
    \big(\pi_B(\mcB)\big)'$ \cite[III.4.5.8]{B17}.

For part (2), since GNS representations of a state are unitarily equivalent, there exists a unitary operator $U$ acting on $H_A\otimes H_B$ such that 
\begin{align}
    &U\ket{\psi}=\ket{\psi'},\text{ and } \label{eq:states}\\
    &U\pi_A(\alpha)\otimes\pi_B(\beta)U^*=\pi_A(\alpha)\otimes\pi_B(\beta)\label{eq:reps} \text{ for all }\alpha\in\mcA,\beta\in\mcB.
\end{align}
Since $\pi_A$ and $\pi_B$ are irreducible, part (1) and \Cref{eq:reps} implies that $U\in\big(\pi_A(\mcA)\otimes\pi_B(\mcB)\big)'=\C\Id$. Then part (2) follows from \Cref{eq:states}.

For part (3), if $\pi_A \otimes \pi_B \iso \wtd{\pi}_A \otimes \wtd{\pi}_B$ as
an $\mcA \otimes \mcB$-module, then $\pi_A \otimes \Id_{H_B}
\iso \wtd{\pi}_A \otimes \Id_{\wtd{H}_B}$ as an $\mcA$-module. But this
requires $\pi_A \iso \wtd{\pi}_A$ by Schur's lemma. 
\end{proof}

We note that \Cref{lemma:irrep} and its proof hold for representations of $C^*$-algebras on infinite-dimensional Hilbert spaces.


Now we are ready to prove the ``if'' direction of \cref{thm:mainresult1}
(with, again, the weaker hypotheses mentioned after the theorem statement).

\begin{theorem}\label{thm:uniquestate}
Let $p\in C_q(X,Y,A,B)$ be an extreme point in $C_q$, let $\mcC$ be a class of
quantum models that is closed, and let $\mcS=\{f_S:S\in\mcC\}$.
If $p$ is an abstract state self-test for $\mcS$, 
then $p$ is a self-test for $\mcC$, and furthermore $p$ has an ideal model
$\wtd{S}$ such that the associated representation
$\wtd{\pi}_A\otimes\wtd{\pi}_B$ is irreducible.
\end{theorem}

\begin{proof} 
Let $f$ be the unique state in $\mcS$ achieving $p$. 
Given a quantum model $S=(H_A,H_B,\{M^x_a\},\{N^y_b\},\ket{\psi})$ in $\mcC$
for $p$ with the associated representation $\pi_A\otimes\pi_B$, we can
decompose the local Hilbert spaces and associated representations with the
double-commutant theorem to get
\begin{align*}
    &H_A=\bigoplus_i H^A_i\otimes K^A_i,\quad \pi_A(\POVM^{X,A})=\bigoplus_i\msB(H^A_i)\otimes\Id_{K^A_i}, \text{ and} \\
    &H_B=\bigoplus_j H^B_j\otimes K^B_j,\quad \pi_B(\POVM^{Y,B})=\bigoplus_j\msB(H^B_j)\otimes\Id_{K^B_j}
\end{align*}
for some Hilbert spaces $H^A_i, K^A_i, H^B_j, K^B_j$. 
 Choose orthonormal bases $\{\ket{\alpha^k_i}\}_{k}$ and $\{\ket{\beta_j^\ell}\}_{\ell}$ for $K_i^A$ and $K_j^B$ respectively. Then
 \begin{align*}
     \ket{\psi}=\bigoplus_{i,j}\left(\sum_{k,\ell}\ket{\psi_{ij}^{kl}}\otimes \lambda_{ij}^{k\ell}\ket{\alpha_i^k,\beta_j^\ell}\right)\in \bigoplus_{i,j} H_i^A\otimes H_j^B\otimes K_i^A\otimes K_j^B,
 \end{align*}
 where $\sum_{i,j}\sum_{k,\ell}\abs{\lambda_{ij}^{k\ell}}^2=1$ and every $\ket{\psi_{ij}^{k\ell}}$ is a vector state in $H_i^A\otimes H_j^B$. Moreover, $\pi_A$ and $\pi_B$ decompose as
 \begin{equation*}
     \pi_A(\alpha)=\bigoplus_{i,k}\pi_i^A(\alpha)\otimes \ket{\alpha_i^k}\bra{\alpha_i^k} \text{ and }
     \pi_B(\beta)=\bigoplus_{j,\ell}\pi_j^B(\beta)\otimes \ket{\beta_j^\ell}\bra{\beta_j^\ell},
 \end{equation*}
 where $\pi_i^A$ and $\pi_j^B$ are irreducible representations.
 
For convenience, let 
$\Lambda:=\{(i,j,k,\ell):\lambda_{ij}^{k\ell}\neq 0\}$ be the set of indices for which the coefficient
$\lambda_{ij}^{k\ell}$ is non-zero. It is clear that $\Lambda$ is
non-empty. Now, for every $(i,j,k,\ell)\in\Lambda$, 
$S_{ij}^{kl}:=(H_i^A,H_j^B,\{\pi_i^A(m^x_a)\},\{\pi_j^B(n^y_b)\},\ket{\psi_{ij}^{k\ell}})$
is a submodel 
of $S$. Since $\mcC$ is closed, 
 we see that $S_{ij}^{k\ell}\in\mcC$.
And if $p_{ij}^{k\ell} \in C_q$ is the correlation of $S_{ij}^{k\ell}$, then 
$\sum_{(i,j,k,\ell)\in\Lambda}\abs{\lambda_{ij}^{k\ell}}^2p_{ij}^{k\ell}=p$. But since $p$
is an extreme point in $C_q$ and $f$ is the unique abstract state in $\mcS$ achieving $p$, 
we must have $p_{ij}^{k\ell}=p$ and ultimately $f_{S_{ij}^{k\ell}}=f$ for all
$(i,j,k,\ell)\in\Lambda$. 

In particular, it follows from the above argument that there are models $$\wtdmodel
\in \mcC$$ for $p$ such that the associated representation $\wtd{\pi}_A \otimes
\wtd{\pi}_B$ is irreducible. Let's suppose that we've fixed such a model
$\wtd{S}$ (not necessarily arising from $S$), and continue working with the
model $S$.  Since $\pi_i^A$ and $\pi_j^B$ are irreducible, if $(i,j,k,\ell) \in
\Lambda$ then
$(\pi^A_i \otimes \pi^B_j, H^A_i \otimes H^B_j, \ket{\psi_{ij}^{k\ell}})$ is a GNS
representation of $f$. Hence by Lemma \ref{lemma:irrep}, part (2), for any
$i,j$ the vector states $\ket{\psi_{ij}^{k\ell}}, (i,j,k,\ell) \in \Lambda$ differ at
most by a complex phase. By absorbing this phase into the coefficients
$\lambda_{ij}^{k\ell}$, we can rewrite
\begin{equation*}
    \ket{\psi}=\bigoplus_{ij}\ket{\psi_{ij}}\otimes\lambda_{ij}\ket{\kappa_{ij}},
\end{equation*}
where $\ket{\psi_{ij}} \in H^A_i \otimes H^B_j$ and $\ket{\kappa_{ij}} \in K^A_i
\otimes K^B_j$ are vector states, and $\sum_{i,j} |\lambda_{ij}|^2 = 1$. 

Since $\wtd{\pi}_A \otimes \wtd{\pi}_B$ is irreducible, $(\wtd{\pi}_A \otimes \wtd{\pi}_B,
\wtd{H}_A \otimes \wtd{H}_B, \ket{\wtd{\psi}})$ is a GNS representation of $f$.
Similarly, $(\pi^A_i \otimes \pi^B_j, H^A_i \otimes H^B_j, \ket{\psi_{ij}})$ is a 
GNS representation of $f$ for $(i,j) \in \Lambda' := \{(i,j) : \lambda_{ij} \neq 0\}$. Hence
the representations $\pi^A_i \otimes \pi^B_j,
(i,j) \in \Lambda'$ are all unitarily equivalent to $\wtd{\pi}_A \otimes \wtd{\pi}_B$. By \Cref{lemma:irrep},
part (3), if $i \in \Lambda_A := \{i : \lambda_{ij} \neq 0 \text{ for some }
j\}$, then $\pi^A_i$ is unitarily equivalent to $\wtd{\pi}_A$ via some unitary
$U^A_i : H^A_i \to \wtd{H}_A$. Similarly if $j \in \Lambda_B := \{j :
\lambda_{ij} \neq 0 \text{ for some } i\}$, then $\pi^B_j$ is unitarily
equivalent to $\wtd{\pi}_B$ via some unitary $U^B_j : H^B_j \to \wtd{H}_B$.
If $(i,j) \in \Lambda'$, then $(\wtd{\pi}_A \otimes \wtd{\pi}_B, \wtd{H}_A
\otimes \wtd{H}_B, U^A_i \otimes U^B_j \ket{\psi_{ij}})$ is a GNS representation
of $f$, so $U^A_i \otimes U^B_j \ket{\psi_{ij}}$ and $\ket{\wtd{\psi}}$ differ
by a complex phase. By absorbing this phase into $\lambda_{ij}$, we can assume
that $U^A_i \otimes U^B_j \ket{\psi_{ij}} = \ket{\wtd{\psi}}$ for all $(i,j) \in \Lambda'$.

For every $i\in \Lambda_A$, we can define an isometry $I_i^A:H_i^A\otimes
K_i^A\to\wtd{H}_A\otimes K_i^A$ via $I_i^A:=U_i^A\otimes \Id_{K_i^A}$. For
every $i\notin \Lambda_A$, let $I_i^A:H_i^A\otimes
K_i^A\arr \wtd{H}_A \otimes(H_i^A\otimes K_i^A)$ be the isometry
$\ket{v}\mapsto\ket{w}\otimes\ket{v}$, where $\ket{w} \in \wtd{H}_A$ is
some arbitrarily chosen vector state.
Define isometries $I_j^B$ analogously. Let
\begin{equation*}
    H_A^{aux}:=\left(\bigoplus_{i\in\Lambda_A}K_i^A\right)\oplus\left(\bigoplus_{i\notin\Lambda_A}(H_i^A\otimes K_i^A) \right) \text{ and } 
    H_B^{aux}:=\left(\bigoplus_{j\in\Lambda_B}K_j^B\right)\oplus\left(\bigoplus_{j\notin\Lambda_B}(H_j^B\otimes K_j^B) \right),
\end{equation*}
and define $\ket{aux}:=\bigoplus_{(i,j) \in \Lambda'}\lambda_{ij}\ket{\kappa_{ij}} \in H_A^{aux} \otimes H_B^{aux}$. Finally,
letting $I_A:=\bigoplus_{i}I_i^A$ and $I_B:=\bigoplus_{j}I_j^B$, we see that
\begin{equation*}
    (I_A \otimes I_B) \pi_A(\alpha) \otimes \pi_B(\beta) \ket{\psi} = \left( \wtd{\pi}_A(\alpha) \otimes \wtd{\pi}_B(\beta) \ket{\wtd{\psi}} \right)
        \otimes \ket{aux}
\end{equation*}
for all $\alpha \in \POVM^{X,A}$, $\beta \in \POVM^{Y,B}$. We conclude that $S \succeq \wtd{S}$, so $p$ is a self-test for $\mcC$
with ideal model $\wtd{S}$. 
\end{proof}
The point in \Cref{thm:uniquestate} of showing that there is an ideal strategy
\begin{equation*}
    \wtdmodel
\end{equation*}
for which the associated representation $\wtd{\pi}_A \otimes \wtd{\pi}_B$ is
irreducible is that then $\ket{\wtd{\psi}}$ is cyclic for $\wtd{\pi}_A \otimes
\wtd{\pi}_B$, so $\left(\wtd{\pi}_A \otimes \wtd{\pi}_B, \wtd{H}_A \otimes
\wtd{H}_B, \ket{\wtd{\psi}}\right)$ is the GNS representation of the unique
abstract state $f = f_{\wtd{S}}$ for $p$. In particular, the GNS representation of
$f$ will be irreducible, which means that $f$ is an extreme
point of the set of states on $\POVM^{X,A} \otimes_{min} \POVM^{Y,B}$. 

As mentioned in \cref{sec:correlations}, the assumption that the correlation
$p$ is an extreme point in $C_q$ is necessary for \Cref{thm:uniquestate}.
In the following example we describe a non-extreme quantum correlation which
is an abstract state self-test but not a self-test.

\begin{example}\label{ex:extremal_cor_not_self-test}
Let $X=Y=\{0\}$ and $A=B=\{0,1\}$, and consider the correlations $p$, $p_1$, and $p_2$ defined by 
\begin{align*}
& p(0,0|0,0)=p(1,1|0,0)=\frac{1}{2},\  p(0,1|0,0)=p(1,0|0,0)=0,\\
& p_1(0,0|0,0)=1,\  p_1(0,1|0,0)=p_1(1,0|0,0)=p_1(1,1|0,0)=0, \text{ and }\\
& p_2(0,0|0,0)=p_2(0,1|0,0)= p_2(1,0|0,0)=0,\ p_2(1,1|0,0)=1. 
\end{align*}
With question sets of size $1$, all correlations are classical, so in particular
$p$, $p_1$, and $p_2$ belong to $C_q(X,Y,A,B)$. Since
$p=\frac{1}{2}p_1+\frac{1}{2}p_2$, $p$ is not an extreme point in $C_q$. 

Suppose $f$ is an abstract state on
$\POVM^{X,A}\otimes_{min}\POVM^{Y,B}$ achieving $p$. 
The correlation $p$ is synchronous, so as we'll show in
\cref{prop:synchronousprojectivestate}, $f$ is projective. The algebra
$\PVM^{X,A}\otimes_{min}\PVM^{Y,B}$ is spanned
by $1\otimes 1$, $m^0_0 \otimes 1$, $1 \otimes n^0_0$, and $m^0_0 \otimes n^0_0$, so $f$ is determined by
its values on these monomials. From the correlation $p$, we see that 
\begin{align*}
    & f(1\otimes 1) = 1, \quad f(m^0_0 \otimes 1) = p(0,0|0,0) + p(0,1|0,0) = \frac{1}{2},   \\
    & f(1 \otimes n^0_0) = p(0,0|0,0) + p(1,0|0,0) = \frac{1}{2}, \quad f(m^0_0 \otimes n^0_0) = p(0,0|0,0) = \frac{1}{2}.
\end{align*}
So $p$ is an abstract state self-test for the set of all states on 
$\POVM^{X,A}\otimes_{min}\POVM^{Y,B}$.

We will now show that $p$ is not a self-test for the class of quantum models. Consider the following quantum models for $p$:
\begin{align*}
&\widehat{S}=\left(\C^2,\C^2, \{\widehat{M}_0=\widehat{N}_0=\ket{0}\bra{0},\widehat{M}_1=\widehat{N}_1=\ket{1}\bra{1}\},\ket{\widehat{\psi}}=\tfrac{\ket{00}+\ket{11}}{\sqrt{2}}\right),
\text{ and}\\
&S=\Big(\C^3,\C^3, \{M_0=N_0=\ket{0}\bra{0},M_1=N_1=\ket{1}\bra{1}+\ket{2}\bra{2}\}, \ket{\psi}=\tfrac{\ket{00}}{\sqrt{2}}+\tfrac{\ket{11}+\ket{22}}{2}\Big).
\end{align*}
Now, assume $p$ is a self-test for quantum models with ideal model 
\begin{equation*}
    \wtdmodel.
\end{equation*}
Then $S\succeq \wtd{S}$ via some isometries $I_A,I_B$ and vector state
$\ket{aux}$, and $\widehat{S}\succeq\wtd{S}$ via some isometries
$\widehat{I}_A,\widehat{I}_B$ and vector state $\ket{\widehat{aux}}$, so
in particular
\begin{align*}
    I_A\otimes I_B \ket{\psi}=\ket{\wtd{\psi}}\otimes \ket{aux} \text{ and }
    \widehat{I}_A\otimes \widehat{I}_B\ket{\widehat{\psi}}=\ket{\wtd{\psi}}\otimes \ket{\widehat{aux}}.
\end{align*}
The states $\ket{\psi}$ and $\ket{\widehat{\psi}}$ have Schmidt rank $2$ and
$3$ respectively. Since isometries preserve Schmidt rank, the Schmidt rank of $\ket{\wtd{\psi}}$ must divide both $3$ and $2$. Therefore the Schmidt rank of
$\ket{\wtd{\psi}}$ is $1$. It's well-known that $p$ cannot be attained with
a separable pure state, so $p$ is not a self-test for quantum models. Indeed,
suppose that
\begin{equation*}
    \model
\end{equation*}
is a model for $p$ with $\ket{\psi} = \ket{\psi_1} \ket{\psi_2}$. Then $p(a,b|0,0) = p_a q_b$,
where $p_a = \braket{\psi_1|M^0_a|\psi_1}$ and $q_b = \braket{\psi_2|N^0_a|\psi_2}$. But there
are no probability distributions $\{p_0,p_1\}$ and $\{q_0,q_1\}$ with $p_0 q_0 = p_1 q_1 = 1/2$ and
$p_0 q_1 = p_1 q_0 = 0$.
\end{example}

\subsection{The classes of all quantum models and projective quantum models}\label{section:mainresult2}

In this section we look closer at the class $\mcC_{POVM}$ of all quantum models
and the class $\mcC_{PVM}$ of projective quantum models. We start by proving
\Cref{cor:mainresult2}. 

\begin{proof}[Proof of \Cref{cor:mainresult2}]
    $\mcC_{POVM}$ and $\mcC_{PVM}$ are both closed. 
    Recall from \Cref{sec:correlations} that if $f$ is a state
    on $\POVM^{X,A} \otimes_{min}\POVM^{Y,B}$, then $f=f_S$ for $S \in \mcC_{POVM}$
    (resp. $\mcC_{PVM}$) if and only if $f$ is finite-dimensional (resp. projective
    and finite-dimensional).
    If $S$ is a model for $p$, then the support model of $S$ is a full-rank
    model for $p$.  Hence every $p \in C_q$ has a full-rank model in
    $\mcC_{POVM}$, so part (a) of \Cref{cor:mainresult2} follows immediately
    from \Cref{thm:mainresult1} applied to $\mcC_{POVM}$. Part (b) is 
    \Cref{thm:mainresult1} applied to $\mcC_{PVM}$. 
\end{proof}
As mentioned previously, every correlation $p \in C_q$ has a projective quantum
model, and hence there is a projective finite-dimensional state on $\POVM^{X,A}
\otimes_{min}\POVM^{Y,B}$ with correlation $p$. If $p$ is an abstract state self-test
for finite-dimensional states, then the unique finite-dimensional state achieving $p$
must be projective. It follows immediately that if $p$ is
an abstract state self-test for all finite-dimensional states, then $p$ is an abstract
state self-test for projective finite-dimensional states. This allows us to prove:
\begin{proposition}\label{prop:alltoprojective}
    If $p \in C_q$ is a self-test for all quantum models, then $p$ is a
    self-test for projective quantum models. Furthermore, $p$ has an ideal 
    model which is full-rank and projective. 
\end{proposition}
The proof uses the following lemma:
\begin{lemma}\label{lem:nonprojstate}
    If $S$ is a full-rank model and $f_S$ is projective, then $S$ is
    projective.  
\end{lemma}
\begin{proof}
    Let $\model$. Since $f_S$ is projective,
    \begin{align*}\label{eq:proj}
        \bra{\psi}\big(M^x_a-(M^x_a)^2\big)\otimes \Id\ket{\psi}=f_S\big(m^x_a-(m^x_a)^2 \big)=0
    \end{align*}
    for all $x\in X,a\in A$. Since $M^x_a-(M^x_a)^2$ is a positive operator and
    $\ket{\psi}$ is full-rank, we conclude that $M^x_a - (M^x_a)^2 = 0$ for all
    $a \in A$, $x \in X$. The same argument shows that $N^y_b$ is a projection for all
    $b \in B$, $y \in Y$. 
\end{proof}
\begin{proof}[Proof of \Cref{prop:alltoprojective}]
    By \Cref{prop:onlyifdirection}, $p$ is an abstract state self-test for
    finite-dimensional states, and hence an abstract state self-test for
    finite-dimensional projective states. In particular, the unique abstract
    state $f$ for $p$ is projective. By \Cref{cor:supportmodel}, $p$ has a
    full-rank ideal model $\wtd{S}$. Since $f_{\wtd{S}} = f$,
    \Cref{lem:nonprojstate} implies that $\wtd{S}$ is projective. 
\end{proof}
One way to show that every correlation in $C_q$ has a projective model
is to start with a POVM model, and construct a projective model using Naimark
dilation. 
This suggests a naive strategy for proving \Cref{prop:alltoprojective}:
show that for every POVM model $S$, there is a projective model $\widehat{S}$
such that $S \succeq \widehat{S}$. Unfortunately, this latter
statement is not true. Indeed, 
let $S$ be a full-rank quantum model which is not projective, and suppose
$S \succeq \widehat{S}$ where $\widehat{S}$ is projective. By \Cref{rmk:Scentral}, 
$f_S = f_{\widehat{S}}$ is projective, and 
\Cref{lem:nonprojstate} implies that $S$ is projective, a contradiction. So
$S$ does not locally dilate to any projective model, and the naive strategy outlined
above cannot be used to prove \Cref{prop:alltoprojective}.

As mentioned in \Cref{sec:correlations}, Baptista, Chen, Kaniewski, Lolck, Man{\v{c}}inska, Gabelgaard Nielsen, and Schmidt
\cite{lifiting23} have constructed a correlation which is a self-test for full-rank
(POVM) quantum models, and a self-test for PVM quantum models, but does not have a full-rank projective quantum model.
By \Cref{prop:alltoprojective}, this correlation is not a self-test for all quantum models. We
do not know whether this correlation is an abstract state self-test for projective finite-dimensional states.  

Lemma \ref{lem:nonprojstate} can also be used to prove:
\begin{proposition}\label{prop:consequences}
    Suppose $p\in C_q(X,Y,A,B)$ is an extreme point in $C_q$ and is a self-test for
    all quantum models. Then $p$ has an ideal model which is full-rank and
    projective, and for which the associated representation is irreducible.
\end{proposition}
The proof uses the following lemma:
\begin{lemma}\label{lem:projectionscommute}
    Two projections $P,Q\in \msB(H)$ commute if and only if $PQP$ is a projection.
\end{lemma}
\begin{proof}
The ``only if" direction is obvious. For the ``if" direction, suppose $PQP$ is a projection, i.e., $PQPQP=PQP$. Let $S:=PQ(\Id-P)$. Then $SS^*=PQ(\Id-P)QP=PQP-PQPQP=0$, which implies $S=0$. Hence $PQ=PQP$. Then $QP=(PQ)^*=(PQP)^*=PQP=PQ$.
\end{proof}
\begin{proof}[Proof of \Cref{prop:consequences}]
    By \Cref{cor:mainresult2}, part (a), $p$ is an abstract state self-test for
    finite-dimensional states. As mentioned above, the unique finite-dimensional
    state $f$ corresponding to $p$ must be projective. By \Cref{thm:uniquestate},
    $p$ has an ideal model 
    \begin{equation*}
        \wtdmodel
    \end{equation*}
    for which the associated representation $\wtd{\pi}_A \otimes \wtd{\pi}_B$
    is irreducible, and in particular  $$\left(\wtd{\pi}_A \otimes \wtd{\pi}_B,
    \wtd{H}_A \otimes \wtd{H}_B, \ket{\wtd{\psi}}\right)$$ is a GNS
    representation of the projective state $f_{\wtd{S}} = f$. Since $f$
    is the pullback of a state on $\PVM^{X,A} \otimes_{min} \PVM^{Y,B}$, any
    GNS representation of $f$ is a pullback of a representation of 
    $\PVM^{X,A} \otimes_{min} \PVM^{Y,B}$, and hence $\wtd{S}$ is projective. 

    Let $S'$ be the support model of $\wtd{S}$. By \Cref{lem:nonprojstate}, $S'$
    is projective, so by \Cref{lem:projectionscommute}
    the support projection $\Pi_A$ commutes with $\wtd{M}^x_a$ for all $x \in
    X$, $a \in A$. Thus $\Pi_A \in \wtd{\pi}_A(\POVM^{X,A})'$, and since
    $\wtd{\pi}_A$ is irreducible, we must have $\Pi_A = \Id_{\wtd{H}_A}$ is the identity. 
    Similarly $\Pi_B = \Id_{\wtd{H}_B}$, and we conclude that $\wtd{S} = S'$
    is full-rank.
\end{proof}
\begin{remark}
    More generally, the argument in the last paragraph of the proof shows that if
    $p$ is a self-test for all quantum models, and an extreme point of $C_q$, then
    \begin{enumerate}[(1)]
        \item any full-rank model for $p$ is projective, and
        \item any projective model for $p$ is centrally supported. 
    \end{enumerate}
\end{remark}

\section{Self-testing for synchronous and binary correlations}\label{sec:syn}

In this section we prove \cref{thm:mainresult3}.

\subsection{Synchronous correlations}

The set of synchronous quantum correlations $C_q^{syn}(X,A)$ is the
convex subset of $C_q(X,X,A,A)$ such that $p(a,b|x,x)=0$ for all $x \in X$
and $a \neq b \in A$. The following result is well-known to experts (see, e.g.
\cite[Theorem 5.5]{PSSTW16} and \cite[Lemma 3]{CM14}). We provide a proof for
the convenience of the reader. 
\begin{lemma}\label{lem:synch}
    Let $p\in C_q^{syn}(X,A)$. If $S=(H_A\otimes H_B,\{M^x_a:a\in A, x\in X\},\{N^y_b:b\in A,y\in X\}, \ket{\psi})$ is a quantum model for $p$, then $M^x_a\otimes \Id\ket{\psi}=\Id\otimes N^x_a\ket{\psi}$ for all $x\in X$, $a\in A$. If in addition the vector state $\ket{\psi}$ has full-rank, then $S$ is a projective quantum model.
\end{lemma}
\begin{proof}
    By Naimark's dilation theorem there are isometries $V_A,V_B$ and
    self-adjoint projections $P^x_a,Q^y_b$, such that
    $M^x_a=V_A^*P^x_aV_A,Q^y_b=V_B^*N^y_b V_B$. If we define
    $\ket{\psi'}:=V_A\otimes V_B\ket{\psi}$, then
    \begin{equation*}
        S':=(V_AH_A,V_BH_B,\{P^x_a:a\in A,x\in X\},\{Q^y_b:b\in A, y\in X\},\ket{\psi'})
    \end{equation*}
    is a projective quantum model for $p$. By \cite[Theorem 5.5]{PSSTW16}, it
    follows that for each $a\in A$ and $x\in X$ we have $P^x_a\otimes
    \Id\ket{\psi'}=\Id\otimes Q^x_a\ket{\psi'}$, and therefore $P^x_aV_A\otimes
    V_B\ket{\psi}=V_A\otimes Q_a^x V_B\ket{\psi}$. Applying the adjoint
    $V_A^*\otimes V_B^*$ to both sides, we obtain $M^x_a\otimes
    \Id\ket{\psi}=\Id\otimes N^x_a\ket{\psi}$ for all $a\in A$, $x\in X$. From this
    observation we can deduce that 
    \begin{align*}
        \bra{\psi} (M^x_a)^2\otimes \Id\ket{\psi}=\bra{\psi} M^x_a\otimes N^x_a\ket{\psi}=\sum_{b\in A}\bra{\psi}M^x_a\otimes N^x_b\ket{\psi}=\bra{\psi} M^x_a\otimes \Id\ket{\psi},
    \end{align*}
    and therefore $\bra{\psi} \big(M^x_a-(M^x_a)^2\big)\otimes
    \Id\ket{\psi}=0$ for all $a\in A$, $x\in X$. Since $M^x_a-(M^x_a)^2$ is a
    positive operator, if $\ket{\psi}$ is full-rank then $(M^x_a)^2=M^x_a$ for
    all $a\in A$, $x\in X$, and hence $S$ is a projective quantum model for
    $p$.
\end{proof}
When combined with \Cref{prop:centrallysupported}, \Cref{lem:synch} implies that
every quantum model for a synchronous correlation is centrally supported. Hence:
\begin{proposition}\label{prop:synchronousmain3}
    A synchronous quantum correlation $p\in C_q^{syn}$ is a self-test for
    quantum models if and only if $p$ is a self-test for projective quantum
    models.
\end{proposition}
\begin{proof}
    Suppose $p \in C_q^{syn}$ is a self-test for projective quantum models with
    ideal model $\wtd{S}$. If
    $S$ is a quantum model for $p$, then $S$ is centrally supported
    by \Cref{lem:synch} and \Cref{prop:centrallysupported}. Hence if $S'$ is
    the support model for $S$, then $S \succeq S'$ by \Cref{lem:supportmodel}.
    Since $S'$ is full-rank, $S'$ is projective by \Cref{lem:synch}, and hence
    $S' \succeq \wtd{S}$. So $p$ is a self-test for all quantum models. 

    The other direction follows from \Cref{prop:alltoprojective}.
\end{proof}

The same result holds for abstract state self-tests:
\begin{proposition}\label{prop:synchronousprojectivestate}
    If $p \in C_q^{syn}$ and $S$ is a quantum model for $p$, then $f_S$ is projective.
    Hence $p$ is an abstract state self-test for finite-dimensional states if and only if
    $p$ is an abstract state self-test for projective finite-dimensional states. 
\end{proposition}
\begin{proof}
    By \Cref{lem:synch} and \Cref{prop:centrallysupported}, every model $S$ for $p$ is
    centrally supported. If $S'$ is the support model for $S$, then $S'$ is projective, 
    and $f_S = f_{S'}$ is projective. 
\end{proof}

\subsection{Binary correlations}

A correlation is said to be a binary correlation if the measurement
scenario has only two outputs, usually taken to be $A = B = \{0,1\}$.
We adapt a proof technique from \cite[Proposition 2]{CHTW10}:

\begin{lemma}\label{lem:binary}
Suppose $p\in C_q(X,Y,\{0,1\},\{0,1\})$ is an extreme point in $C_q$. Let \begin{equation*}
    S=(H_A,H_B,\{M_0^x,M_1^x:x\in X\},\{N_0^y,N_1^y:y\in Y\},\ket{\psi})
\end{equation*}
be a quantum model of $p$. For each $x\in X$, let
$\sum_{i=1}^{d}\lambda_i^x\ket{\phi_i^x}\bra{\phi_i^x}$ be the spectral
decomposition of $M^x_0$, where $d=dim(H_A)$. Then one of the following conditions must hold
for each $1\leq i \leq d$ and $x\in X$: 
\begin{enumerate}[(a)]
    \item $\lambda_i^x=0$,
    \item $\lambda_i^x=1$, or
    \item $\ket{\phi_i^x}\bra{\phi_i^x}\otimes\Id\ket{\psi}=0$. 
\end{enumerate}
\end{lemma}
\begin{proof}
    Suppose $x_0\in X$ and $i_0 \in \{1,\ldots,d\}$. Define new binary outcome
    measurements $\{E^{x_0}_0,E^{x_0}_1\}$ and $\{F^{x_0}_0, F^{x_0}_1\}$ by 
    $E^{x_0}_0:=\sum_{j\neq i_0}\lambda_j^{x_0}\ket{\phi_j^{x_0}}\bra{\phi_j^{x_0}}$,
    $E^{x_0}_1 := \Id - E^{x_0}_0$,
    $F^{x_0}_0:=E^{x_0}_0+\ket{\phi_{i_0}^{x_0}}\bra{\phi_{i_0}^{x_0}}$, and $F^{x_0}_1:=\Id-F^{x_0}_0$. 
    For the remaining $x\neq x_0\in X$, set 
    $E^{x}_a=F^{x}_a :=M^{x}_a$, so that
    \begin{equation}\label{eq:comb}
       M^x_a=\lambda_{i_0}^{x_0} F^{x}_a+(1-\lambda_{i_0}^{x_0})E^x_a
    \end{equation} 
    for all $x\in X$, $a\in\{0,1\}$. Let $p_1$ and $p_2$ be the correlations with quantum models
    $S_1:=(H_A,H_B,\{F^x_a:a\in \{0,1\}, x\in X\},\{N^y_b:b\in \{0,1\},y\in
    Y\},\ket{\psi})$ and $S_2:=(H_A,H_B,\{E^x_a:a\in \{0,1\}, x\in X\},\{N^y_b:b\in
    \{0,1\},y\in Y\},\ket{\psi})$ respectively.  By
    \Cref{eq:comb} we see that $p=\lambda_{i_0}^{x_0}p_1+(1-\lambda_{i_0}^{x_0})p_2$. Since $p$
    is an extreme point in $C_q$, either $\lambda_{i_0}^{x_0}=0$, $\lambda_{i_0}^{x_0}=1$, or
    $p_1 = p_2 = p$. In the last case, we have that
    $0=\bra{\psi}(F^{x_0}_0-E^{x_0}_0)\otimes\Id\ket{\psi}=\bra{\psi}(\ket{\phi_{i_0}^{x_0}}\bra{\phi_{i_0}^{x_0}})\otimes
    \Id\ket{\psi}$, and since $\ket{\phi_{i_0}^{x_0}}\bra{\phi_{i_0}^{x_0}}\otimes \Id$ is a positive
    operator, $(\ket{\phi_{i_0}^{x_0}}\bra{\phi_{i_0}^{x_0}}\otimes\Id)\ket{\psi}=0$. 
\end{proof}

\begin{proposition}\label{prop:binary}
    Suppose $p\in C_q(X,Y,\{0,1\},\{0,1\})$ is an extreme point in $C_q$. If
    $S$ is a quantum model for $p$, then there is a projective quantum model $S'$
    such that $S \succeq S'$.
\end{proposition}
\begin{proof}
    Suppose $\model$. Once again let
    $\sum_{i}\lambda^x_i\ket{\phi^x_i}\bra{\phi^x_i}$ be the spectral
    decomposition of $M^x_0$ for each $x \in X$. For all $x$ and $i$, let
    \begin{equation*}
        \wtd{\lambda}^x_i:=\begin{cases} 1 & \lambda^x_{i}=1,\\
            0 & \text{ otherwise,} \end{cases}
    \end{equation*}
    and define $P^x_0:=\sum_i\wtd{\lambda}^x_i\ket{\phi^x_i}\bra{\phi^x_i}$ and 
    $P^x_1=\Id-P^x_0$. By \Cref{lem:binary}, $\{P^x_0,P^x_1\}$ is a binary outcome PVM such that
    \begin{align*}
        M^x_a\otimes \Id\ket{\psi}=P^x_a\otimes \Id \ket{\psi}
    \end{align*}
    for $a \in \{0,1\}$. Similarly, there are binary outcome PVMs $\{Q^y_0,Q^y_1\}$ on $H_B$ such
    that 
    \begin{equation*}
        \Id \otimes N^y_b \ket{\psi} = \Id \otimes Q^y_b\ket{\psi}
    \end{equation*}
    for all $y \in Y$, $b \in \{0,1\}$. If  $S'=(H_A,H_B,\{P_a^x:a\in A,x\in
X\},\{Q_b^y:b\in \{0,1\},y\in Y\},\ket{\psi})$, then $S \succeq S'$ by the identity
    isometries. 
\end{proof}

\begin{corollary}\label{cor:binaryequiv}
    If $p$ is a binary correlation and an extreme point of $C_q$, then $p$ is a self-test
    for quantum models if and only if $p$ is a self-test for projective quantum models.
\end{corollary}

\begin{proof}
    By \Cref{prop:alltoprojective}, if $p$ is a self-test for all quantum
    models then it is a self-test for projective quantum models. For the other
    direction, suppose $p$ is a self-test for the class of projective quantum
    models with ideal model $\wtd{S}$. By \Cref{prop:binary}, if $S$ is a
    POVM quantum model for $p$, then there is a projective quantum model
    $S'$ with $S \succeq S'$. Since $S' \succeq \wtd{S}$, $p$ is a self-test for all quantum models.
\end{proof}

\subsection{Proof of \cref{thm:mainresult3}}
If $p$ is a synchronous correlation, then the equivalence of (1) and (2) is
\Cref{prop:synchronousmain3}, while the equivalence of (3) and (4) is
\Cref{prop:synchronousprojectivestate}. If $p$ is an extreme point, then
(1) and (3) are equivalent by \Cref{cor:mainresult2}, part (a).

If $p$ is a binary correlation and an extreme point, then (1) and (2) are
equivalent by \Cref{cor:binaryequiv}, while (1) and (3) are equivalent by
\Cref{cor:mainresult2}, part (a) again. Finally (3) always implies (4), and
(4) implies (2) by \Cref{thm:uniquestate}.

\section{Self-testing for infinite-dimensional tensor product models}\label{sec:tensor_product_models}

Recall that a \textbf{tensor product (POVM) model} for a correlation $p$ is a tuple
\begin{equation*}
    \model,
\end{equation*}
where the Hilbert spaces $H_A$ and $H_B$ are not restricted to be
finite-dimensional. The set of correlations $p$ with a tensor product model is
denoted by $C_{qs}=C_{qs}(X,Y,A,B)$. Like $C_{q}$, the set $C_{qs}$ is convex
but not closed \cite{Slof19b}; however, its closure is also the set $C_{qa}$.
Like $C_q$, correlations $p\in C_{qs}$ can be expressed by abstract states
on $\POVM^{X,A} \otimes_{min} \POVM^{Y,B}$. We say a state $f$ on $\POVM^{X,A}
\otimes_{min} \POVM^{Y,B}$ is a \textbf{tensor product state} if there exists a
tensor product model $S$ such that $f=f_S$. Because of the tensor product
structure, the standard definition of self-testing using local dilations in
\Cref{def:localdilation} and \Cref{def:tensorproductselftest} extends straightforwardly
to tensor product models and correlations in $C_{qs}$. 
Coladangelo and Stark have constructed examples of correlations which are in
$C_{qs}$ but not in $C_q$ \cite{CS18}. We do not know of any examples of
self-tests in $C_{qs} \setminus C_{q}$. Since the double-commutant theorem does
not hold for infinite-dimensional Hilbert spaces, we also do not know whether
abstract state self-tests are self-tests when $p$ is an extreme point. 
However, we can prove that self-tests for correlations in $C_{qs}$ are abstract
state self-tests:
\begin{proposition}\label{thm:qs_abstract_state_self-test}
If $p\in C_{qs}$ is a self-test for the class of tensor product models then $p$ is an abstract state self-test for the subset of tensor product states.
\end{proposition}
The proof of \cref{thm:qs_abstract_state_self-test} is similar to the proof of
\Cref{prop:onlyifdirection}. First, note that the definition of support
projections, support models, and centrally supported models generalize straightforwardly
to tensor product models. A tensor product model $S$ is said to be full-rank if
the support projections of the vector state in $S$ are the identity operators,
so in particular a full-rank tensor product model is centrally supported. Although
\cref{lemma:onlyif} and \cref{lemma:if} no longer hold for tensor product
models, they can be replaced by ``approximate'' versions. To state these
lemmas, we use the following notation: if $\ket{v_1}$ and $\ket{v_2}$
are vectors in a Hilbert space $H$, we write
$\ket{v_1}\approx_\epsilon\ket{v_2}$ if $\norm{\ket{v_1}-\ket{v_2}}\leq
\epsilon$. By the triangle inequality, we see that
$\ket{v_1}\approx_{\epsilon_1}\ket{v_2}\approx_{\epsilon_2}\ket{v_3}$ implies
$\ket{v_1}\approx_{\epsilon_1+\epsilon_2}\ket{v_3}$. Also, if $\ket{v_1} \approx_{\eps} \ket{v_2}$ and $T \in \msB(H)$, then $T \ket{v_1} \approx_{\norm{T} \eps} T \ket{v_2}$. 

\begin{lemma}\label{lemma:qs-if}
Let $\ket{\psi}\in H_A\otimes H_B$ be a bipartite vector state, let $\Pi_A$ and $\Pi_B$ be the support projections of $\ket{\psi}$, and let $E$ be a self-adjoint operator in $\msB(H_A)$. If, for any $\epsilon>0$, there exists an operator $\widehat{E}\in\msB(H_B)$ such that $E\otimes \Id\ket{\psi}\approx_\epsilon\Id\otimes \widehat{E}\ket{\psi}$, then the support projection $\Pi_A$ commutes with $E$.
\end{lemma}
\begin{proof}
    Let $\ket{\psi}=\sum_{i\in\Lambda}\lambda_i\ket{\alpha_i}\otimes\ket{\beta_i}$ be a
    Schmidt decomposition of $\ket{\psi}$ where $\Lambda\subset \N$, so the support
    projection of $\ket{\psi}$ on $H_A$ is $\Pi_A=\sum_{i\in
    \Lambda}\ket{\alpha_i}\bra{\alpha_i}$. Now, fix an $\epsilon>0$ and let
    $\widehat{E}$ be an operator in $\mathscr{B}(H_B)$ such that $E\otimes
    \Id\ket{\psi}\approx_{\epsilon/2}\Id\otimes \widehat{E}\ket{\psi}$. Since
    $\Pi_A\leq \Id$, we see that
    \begin{align*}
        \Pi_AE\otimes \Id\ket{\psi}\approx_{\epsilon/2}\Pi_A\otimes \widehat{E}\ket{\psi}=\Id\otimes \widehat{E}\ket{\psi}\approx_{\epsilon/2}E\otimes \Id\ket{\psi}.
    \end{align*}
    Hence $(\Id-\Pi_A)E\otimes \Id\ket{\psi}\approx_{\epsilon} 0$ for all
    $\epsilon>0$, and we must have $(\Id-\Pi_A)E\otimes \Id\ket{\psi}= 0$. This
    implies
    \begin{align*}
        \sum_{i\in\Lambda}\lambda_i\big((\Id-\Pi_A)E\ket{\alpha_i} \big)\otimes \ket{\beta_i}=0.
    \end{align*}
    Since $\lambda_i>0$ for all $i\in\Lambda$ and $\{\ket{\beta_i}:i\in\Lambda \}$
    is an orthonormal subset, it follows that $(\Id-\Pi_A)E\ket{\alpha_i}=0$ for
    all $i\in\Lambda$. Then
    \begin{equation*}
      (\Id-\Pi_A)E\Pi_A=\sum_{i\in\Lambda}(\Id-\Pi_A)E\ket{\alpha_i}\bra{\alpha_i}=0,
    \end{equation*} 
    and $E\Pi_A=\Pi_AE\Pi_A$. Since $E$ is self-adjoint, $\Pi_AE=(E\Pi_A)^*=\Pi_AE\Pi_A=E\Pi_A$. 
\end{proof}

\begin{lemma}\label{lemma:qs-onlyif}
If $\ket{\psi}\in H_A\otimes H_B$ is a full-rank vector state, then for any $\epsilon>0$ and $E\in\msB(H_A)$ (resp. $F\in\msB(H_B)$), there is an operator $\widehat{E}\in\msB(H_B)$ (resp. $\widehat{F}\in\msB(H_A)$) such that $E\otimes\Id\ket{\psi}\approx_\epsilon\Id\otimes\widehat{E}\ket{\psi}$ (resp. $\Id\otimes F\ket{\psi}\approx_\epsilon\widehat{F}\otimes \Id\ket{\psi}$).
\end{lemma}
\begin{proof}
If $H_A$ or $H_B$ is finite-dimensional, then both are finite-dimensional, and the lemma follows from
\Cref{lemma:if}. Suppose that $H_A$ and $H_B$ are infinite-dimensional Hilbert spaces. 
Since the vector state
$\ket{\psi}\in H_A\otimes H_B$ is full rank, it has Schmidt decomposition
\begin{align*}
\ket{\psi}=\sum_{i=1}^{\infty}\lambda_i\ket{\alpha_i}\otimes\ket{\beta_i},
\end{align*}
where $(\lambda_i)_{i=1}^{\infty}$ is a positive sequence in $\ell^2(\N)$ such that $\sum_{i=1}^\infty\abs{\lambda_i}^2=1$, and $\{\ket{\alpha_i}:i\in\N\}$ and $\{\ket{\beta_i}:i\in\N \}$ are orthonormal bases for $H_A$ and $H_B$ respectively. In particular, $H_A$ and $H_B$ are separable. Suppose $\epsilon>0$ and $E$ is an operator in $\msB(H_A)$. Let $E_{ij}:=\bra{\alpha_i}E\ket{\alpha_j}$ for every $i,j\in\N$. Since $E\otimes\Id\ket{\psi}=\sum_{i,j=1}^{\infty}\lambda_jE_{ij}\ket{\alpha_i}\otimes\ket{\beta_j}$ and $\sum_{i,j=1}^{\infty}\abs{\lambda_jE_{ij}}^2=\norm{E\otimes\Id\ket{\psi}}^2<\infty$, there exists an $N\in \N$ such that
\begin{align*}
  \sum_{i,j=1}^{N}\abs{\lambda_jE_{ij}}^2\geq \norm{E\otimes\Id\ket{\psi}}^2-\epsilon^2.  \label{eq:truncation}
\end{align*}
Let $\wtd{H}_A:=\text{span}\{\ket{\alpha_i}:1\leq i\leq N\}$,
$\wtd{H}_B:=\text{span}\{\ket{\beta_i}:1\leq i\leq N\}$, and let
$\wtd{\Pi}_A:=\sum_{i=1}^N\ket{\alpha_i}\bra{\alpha_i}\in\msB(H_A)$ and
$\wtd{\Pi}_B:=\sum_{i=1}^N\ket{\beta_i}\bra{\beta_i}\in\msB(H_B)$ be the projections
onto $\wtd{H}_A$ and $\wtd{H}_B$ respectively. Then
$\ket{\wtd{\psi}}:=\wtd{\Pi}_A\otimes
\wtd{\Pi}_B\ket{\psi}=\sum_{i=1}^N\lambda_i\ket{\alpha_i}\otimes\ket{\beta_i}$ is a
full-rank vector in the finite-dimensional space $\wtd{H}_A\otimes \wtd{H}_B$,
and
\begin{equation*}
    E\otimes\Id_{H_B}\ket{\psi}\approx_{\epsilon}  \wtd{\Pi}_AE\wtd{\Pi}_A\otimes
    \Id_{\wtd{H}_B}\ket{\wtd{\psi}}.
\end{equation*}
By \Cref{lemma:if}, there exists an operator 
$\widehat{E}\in\msB(\wtd{H}_B)$ such that 
\begin{align*}
\wtd{\Pi}_AE\wtd{\Pi}_A\otimes \Id_{\wtd{H}_B}\ket{\wtd{\psi}}
=\Id_{\wtd{H}_A}\otimes \widehat{E}\ket{\wtd{\psi}}.
\end{align*}
Extending $\widehat{E}$ to $H_B$ by having it act trivially on 
$\{\ket{\beta_i}:i\geq N+1 \}$, we see that
$E \otimes \Id \ket{\psi} \approx_{\epsilon} \Id \otimes \widehat{E} \ket{\psi}$. 
The existence of $\widehat{F}$ for $F$ and $\epsilon$ follows similarly.
\end{proof}

\cref{lemma:qs-if} and \cref{lemma:qs-onlyif} allow us to establish the following characterization of centrally supported tensor product models. 

\begin{proposition}\label{prop:qs-centrallysupported}
    A tensor product model $$\model$$ is centrally supported if and only if for any $\epsilon>0$ and every $a\in A$, $x\in X$ (resp. $b\in B$, $y\in Y$) there exists an operator $\widehat{M}_a^x\in \msB(H_B)$ (resp. $\widehat{N}_b^y\in \msB(H_A)$) such that $M_a^x\otimes \Id|\psi\rangle\approx_\epsilon\Id\otimes \widehat{M}_a^x|\psi\rangle$ (resp. $\Id \otimes N_b^y\ket{\psi}\approx_\epsilon\widehat{N}_b^y\otimes \Id\ket{\psi}$) .
\end{proposition}
\begin{proof}
The ``if" part follows directly from \cref{lemma:qs-if}. For the ``only if" part,
let $\Pi_A$ and $\Pi_B$ be the support projections of $\ket{\psi}$, and suppose
$[\Pi_A,M^x_a]=[\Pi_B,N^y_b]=0$ for all $(a,b,x,y)\in A\times B\times X\times
Y$. Since $\Pi_A\otimes\Pi_B\ket{\psi}$ is full-rank in $\Pi_AH_A\otimes
\Pi_BH_B$, by \cref{lemma:qs-onlyif}, for any $\epsilon>0$ there is an operator
$\widehat{M}^x_a$ in $\msB(\Pi_BH_B)$ (which can be regarded as an operator in $\msB(H_B)$
acting trivially on $(\Id-\Pi_B)H_B$) such that
$\Pi_AM^x_a\Pi_A\otimes\Pi_B\ket{\psi}\approx_\epsilon\Pi_A\otimes
\widehat{M}^x_a\Pi_B\ket{\psi}$ for any $a\in A,x\in X$. Therefore
\begin{align*}
M^x_a\otimes\Id\ket{\psi}=\Pi_AM^x_a\Pi_A\otimes\Pi_B\ket{\psi}\approx_\epsilon\Pi_A\otimes \widehat{M}^x_a\Pi_B\ket{\psi}=\Id\otimes\widehat{M}^x_a\ket{\psi}.
\end{align*}
Similarly, for any $\epsilon>0$ and every $b\in B,y\in Y$ there exists $\widehat{N}_b^y\in \mathcal{B}(H_A)$ such that $\Id \otimes N_b^y\ket{\psi}\approx_\epsilon\widehat{N}_b^y\otimes \Id\ket{\psi}$.
\end{proof}

\begin{proposition}\label{qs-prop:monomial}
Let 
\begin{align*}
   &\model \text{ and }\\
   &\wtdmodel
\end{align*}
be two tensor product models for a correlation $p\in C_{qs}(X,Y,A,B)$, with associated representations $\pi_A\otimes\pi_B$ and $\wtd{\pi}_A\otimes\wtd{\pi}_B$ respectively. Suppose $S\succeq\widetilde{S}$ via local isometries $I_A,I_B$ and vector state $\ket{aux}\in H_A^{aux}\otimes H_B^{aux}$. If $\wtd{S}$ is centrally supported, then
\begin{align}
(\wtd{\pi}_A(\alpha)\otimes\Id_{\wtd{H}_B}\ket{\wtd{\psi}})\otimes\ket{aux}=I_A\pi(\alpha)I_A^*\otimes\Id_{\wtd{H}_B\otimes H_B^{aux}}\ket{\wtd{\psi},aux}\text{ for all }\alpha\in\POVM^{X,A}, \label{eq:qs-monomial}\\
(\Id_{\wtd{H}_A}\otimes\wtd{\pi}_B(\beta)\ket{\wtd{\psi}})\otimes\ket{aux}=\Id_{\wtd{H}_A\otimes H_A^{aux}}\otimes I_B\pi(\beta)I_B^*\ket{\wtd{\psi},aux} \text{ for all }\beta\in\POVM^{Y,B}, \label{eq:qs-monomial2}
\end{align}
and $f_S=f_{\wtd{S}}$.
\end{proposition}

\begin{proof}
As in the proof of \Cref{prop:monomial}, we can show that 
\Cref{eq:qs-monomial} holds for all monomials $\alpha=\alpha_1\cdots
\alpha_k\in \POVM^{X,A}$, 
$\alpha_i\in\{e^x_a:x\in X,a\in A \}$ using induction on the monomial degree $k \in \N$. 
Indeed, the cases $k=0,1$ follow straight from the definition of local dilation of tensor product models.
Suppose that \Cref{eq:monomial} holds for all monomials in $\POVM^{X,A}$ of
degree $k$. For any monomial $\alpha=\alpha_1\cdots\alpha_k\alpha_{k+1}$ in
$\POVM^{X,A}$ of degree $k+1$, let $\alpha':=\alpha_1\cdots\alpha_k$. 
Since $\wtd{S}$ is centrally supported, \Cref{prop:qs-centrallysupported} implies that for any $\eps > 0$,
there is an
$\wtd{F}\in\msB(\wtd{H}_B)$ such that
$\wtd{\pi}_A(\alpha_{k+1}) \otimes\Id\ket{\wtd{\psi}}\approx_{\epsilon/2}\Id\otimes\wtd{F}\ket{\wtd{\psi}}$. 
Note that $\norm{\wtd{\pi}_A(\alpha')}\leq 1$ and $\norm{I_A\pi_A(\alpha')I_A^*}\leq 1$. Thus by the inductive hypothesis,
\begin{align*}
   \big(\wtd{\pi}_A(\alpha)\otimes\Id_{\wtd{H}_B}\ket{\wtd{\psi}}\big)\otimes\ket{aux} 
    & = \big(\wtd{\pi}_A(\alpha')\wtd{\pi}_A(\alpha_{k+1})\otimes\Id_{\wtd{H}_B}\ket{\wtd{\psi}}\big)\otimes\ket{aux} \\
    & \approx_{\epsilon/2}\big(\wtd{\pi}_A(\alpha')\otimes \wtd{F}\ket{\wtd{\psi}}\big)\otimes\ket{aux} \\
   & = I_A\pi_A(\alpha')I_A^*\otimes \wtd{F}\otimes\Id_{H_B^{aux}}\ket{\wtd{\psi},aux} \\
   & \approx_{\epsilon/2}\big(I_A\pi_A(\alpha')I_A^*(\wtd{\pi}(\alpha_{k+1})\otimes\Id_{H_A^{aux}})\big)\otimes \Id_{\wtd{H}_B\otimes H_B^{aux}}\ket{\wtd{\psi},aux} \\
   & =\big(I_A\pi_A(\alpha')I_A^*I_A\pi_A(\alpha_{k+1})I_A^*\big)\otimes \Id_{\wtd{H}_B\otimes H_B^{aux}}\ket{\wtd{\psi},aux} \\
   & =I_A\pi_A(\alpha)I_A^*\otimes\Id_{\wtd{H}_B\otimes H_B^{aux}}\ket{\wtd{\psi},aux}.
\end{align*}
Therefore $\big(\wtd{\pi}_A(\alpha)\otimes\Id_{\wtd{H}_B}\ket{\wtd{\psi}}\big)\otimes\ket{aux}\approx_{\epsilon}I_A\pi_A(\alpha)I_A^*\otimes\Id_{\wtd{H}_B\otimes H_B^{aux}}\ket{\wtd{\psi},aux}$ for all $\epsilon>0$, and we conclude that \Cref{eq:qs-monomial} holds for all $\alpha\in \POVM^{X,A}$. The proof of \Cref{eq:qs-monomial2} is similar.
Hence, for any $\alpha\in \POVM^{X,A}$ and $ \beta\in\POVM^{Y,B}$, 
\begin{align*}
    f_{\wtd{S}}(\alpha\otimes\beta)&=\bra{\wtd{\psi}}\wtd{\pi}_A(\alpha)\otimes\wtd{\pi}_B(\beta)\ket{\wtd{\psi}}\\
    &=\bra{\wtd{\psi},aux}\wtd{\pi}_A(\alpha)\otimes\wtd{\pi}_B(\beta)\otimes\Id_{H_A^{aux}\otimes H_B^{aux}}\ket{\wtd{\psi},aux}\\
    &=\bra{\wtd{\psi},aux} I_A\pi_A(\alpha)I_A^*\otimes I_B\pi_B(\beta)I_B^*\ket{\wtd{\psi},aux}\\
    &=\bra{\psi}\pi_A(\alpha)\otimes\pi_B(\beta)\ket{\psi}\\
    &=f_S(\alpha\otimes\beta),
\end{align*}
which implies $f_S=f_{\wtd{S}}$.
\end{proof}

\begin{proposition}\label{prop:qs-3model}
    Let $S$ and $\wtd{S}$ be two tensor product models for $p\in C_{qs}$. If $S$ is centrally supported and $S\succeq\wtd{S}$, then $\wtd{S}$ is centrally supported.
\end{proposition}
\begin{proof}

Let 
\begin{align*}
    &\model \text{ and }\\
    &\wtdmodel
\end{align*}
be two tensor product models with support projections $\Pi_A,\Pi_B$ and $\wtd{\Pi}_A,\wtd{\Pi}_B$ respectively. Suppose $S$ is centrally supported and $S\succeq
\wtd{S}$ with local isometries $I_A,I_B$ and vector state $\ket{aux}\in
H_A^{aux}\otimes H_B^{aux}$. Then
\begin{align*}
    S':=\big(\wtd{H}_A\otimes H_A^{aux},\wtd{H}_B\otimes H_B^{aux},\{\wtd{M}^x_a\otimes\Id_{H_A^{aux}}\},\{\wtd{N}^y_b\otimes\Id_{H_B^{aux}}\},\ket{\wtd{\psi},aux} \big)
\end{align*}
is a tensor product model for $p$ with support projections $\wtd{\Pi}_A\otimes \Pi_A^{aux}$ and $\wtd{\Pi}_B\otimes \Pi_B^{aux}$, where $\Pi_A^{aux}\in\msB(H_A^{aux})$ and $\Pi_B^{aux}\in\msB(H_B^{aux})$ are the support projections of $\ket{aux}$. Since $S$ is centrally supported, by \Cref{prop:qs-centrallysupported}, for every $x\in X,a\in A$, and $\epsilon>0$, there is an operator $\widehat{M}^x_a\in\msB(H_B)$ such that $M^x_a\otimes \Id\ket{\psi}\approx_\epsilon\Id\otimes \widehat{M}^x_a\ket{\psi}$. Since 
    $I_A I_A^* \otimes I_B I_B^* \ket{\wtd{\psi}, aux} = \ket{\wtd{\psi}, aux}$
and $\Id\otimes \Id \geq I_A I_A^* \otimes \Id \geq I_A I_A^* \otimes I_B I_B^*$, we see that $I_A I_A^* \otimes \Id
\ket{\wtd{\psi}, aux} = \ket{\wtd{\psi}, aux}$ and $\norm{I_A\otimes I_B}\leq 1$. Then, $I_B\widehat{M}^x_aI_B^*$ is an operator in $\msB(\wtd{H}_B\otimes H_B^{aux})$ such that
\begin{align*}
    \Id\otimes I_B\widehat{M}^x_aI_B^*\ket{\wtd{\psi},aux}&=I_AI_A^*\otimes I_B\widehat{M}^x_aI_B^*\ket{\wtd{\psi},aux}\\
    &=(I_A\otimes I_B)(\Id\otimes \widehat{M}^x_a\ket{\psi})\\
    &\approx_\epsilon  (I_A\otimes I_B)(M^x_a\otimes \Id\ket{\psi})\\
    &=(\wtd{M}^x_a\otimes\Id_{H_A^{aux}})\otimes \Id\ket{\wtd{\psi},aux}.
\end{align*}
By \Cref{lemma:qs-if} we see that $\wtd{\Pi}_A\otimes\Pi_A^{aux}$ commutes with $\wtd{M}^x_a\otimes\Id_{H_A^{aux}}$, and hence $[\wtd{\Pi}_A,\wtd{M}^x_a]=0$ for all $x\in X,a\in A$. Similarly, $[\wtd{\Pi}_B,\wtd{N}^y_b]=0$ for all $y\in Y,b\in B$. We conclude that $\wtd{S}$ is centrally supported.
\end{proof}

\begin{proof}[Proof of \cref{thm:qs_abstract_state_self-test}]
Suppose $p\in C_{qs}(X,Y,A,B)$ is a self-test for the class of tensor product models. Let $S$ be a full-rank tensor product model for $p$ with ideal model $\wtd{S}$, so that $S\succeq \wtd{S}$. Full-rank models are centrally supported, so \cref{prop:qs-3model} implies that $\wtd{S}$ is centrally supported. It follows from \cref{qs-prop:monomial} that $f_{\wtd{S}}$ is the unique abstract state on $\POVM^{X,A}\otimes_{min}\POVM^{Y,B}$ achieving $p$.
\end{proof}

\Cref{thm:qs_abstract_state_self-test} applies to the set of states which have
a tensor product model.  There are also states on $\POVM^{X,A} \otimes_{min}
\POVM^{Y,B}$ which do not arise from a tensor product model. Suppose $p$ is a 
correlation in $C_{qs}$, and consider the set of states $f$ on $\POVM^{X,A}
\otimes_{min} \POVM^{Y,B}$ such that $f(m^x_a \otimes n^y_b) = p(a,b|x,y)$ for
all $(a,b,x,y) \in A \times B \times X \times Y$. We do not know whether every
such state has a tensor product model, or whether being a self-test for tensor-product
models implies that there is a unique such state on $\POVM^{X,A} \otimes_{min}
\POVM^{Y,B}$.

\section{Self-testing for commuting operator models}\label{sec:commuting_operator_models}

Recall that a \textbf{commuting operator POVM model} for a correlation $p$ is a tuple
\begin{equation*}
    (H, \{M^x_a : a \in A, x \in X\}, \{N^y_b:b \in B, y \in Y\}, \ket{\psi}),
\end{equation*}
where
\begin{enumerate}[(i)]
    \item $H$ is a Hilbert space,
    \item $\{M^x_a : a \in A\}$, $x \in X$ and $\{N^y_b : b \in B \}$, $y \in Y$ are POVMs on $H$ such that
    \begin{equation*}
        M^x_a N^y_b = N^y_b M^x_a
    \end{equation*}
     for all $(a,b,x,y) \in A \times B \times X \times Y$, and
    \item $\ket{\psi} \in H$ is a vector state 
\end{enumerate}
such that
\begin{equation*}
    p(a,b|x,y) = \braket{\psi| M^x_a \cdot N^y_b|\psi}
\end{equation*}
for all $(a,b,x,y) \in A \times B \times X \times Y$. As with tensor product
models, a commuting operator model is \textbf{projective} if the measurements
$\{M^x_a\}$ and $\{N^y_b\}$ are projective, and \textbf{finite-dimensional} if
$H$ is finite-dimensional. We let $C_{qc} = C_{qc}(X,Y,A,B)$ be the set of
correlations with a commuting operator model. It is well-known that $C_{qc}$ is
closed, convex, and contains $C_{qa}$, that every correlation in $C_{qc}$ has a
projective commuting operator model, and that the set of correlations with
finite-dimensional commuting operator models is equal to $C_q$.

Analogously to the finite-dimensional case, if $\qcmodel$ is a commuting
operator model for a correlation $p \in C_{qc}$, then there is an
\textbf{associated representation} $\phi$ of  $\POVM^{X,A}
\otimes_{max} \POVM^{Y,B}$ with $\phi(m^x_a \otimes n^y_b) =
M^x_a \cdot N^y_b$, and the associated abstract state $f_S$ on $\POVM^{X,A}
\otimes_{max} \POVM^{Y,B}$ defined by $f_S(x) := \braket{\psi|\phi(x)|\psi}$ satisfies $p(a,b|x,y) = f_S(m^x_a \otimes n^y_b)$. The
tuple $(\phi, H, \ket{\psi})$ is a GNS representation
of $f_S$ if and only if $\ket{\psi}$ is a cyclic vector for $\phi$, in which case we say that $S$ is a \textbf{cyclic model}.
 Conversely, taking the GNS representation of a state $f$ on $\POVM^{X,A}
\otimes_{max} \POVM^{Y,B}$ gives a commuting operator model $S$ with $f = f_S$.
Hence a correlation $p \in \R_{\geq 0}^{A \times B \times X \times Y}$ belongs
to $C_{qc}(X,Y,A,B)$ if and only if there is a state $f$ on $\POVM^{X,A}
\otimes_{max} \POVM^{Y,B}$ with $p(a,b|x,y) = f(m^x_a \otimes
n^y_b)$ for all $(a,b,x,y) \in A \times B \times X \times Y$ \cite{Fri12,JNPPSW11}.

The definition of an abstract state self-test on $\POVM^{X,A} \otimes_{min}
\POVM^{Y,B}$ suggests a notion of self-testing for commuting operator models:
\begin{definition}\label{def:co_abstract_self-test}
    Let $\mcS$ be a subset of states on $\POVM^{X,A} \otimes_{max}
    \POVM^{Y,B}$. A correlation $p$ is an \textbf{abstract state self-test for
    $\mcS$} if there exists a unique abstract state $f \in \mcS$ with correlation
    $p$.  We say that $p$ is a \textbf{commuting operator self-test} if it is an
    abstract state self-test for all states on $\POVM^{X,A} \otimes_{max}
    \POVM^{Y,B}$.
\end{definition}
 There is a surjective homomorphism $\POVM^{X,A}\otimes_{max} \POVM^{Y,B} \to
\POVM^{X,A}\otimes_{min} \POVM^{Y,B}$, and this means that states on
$\POVM^{X,A}\otimes_{min} \POVM^{Y,B}$ can be thought of as a subset of states
on $\POVM^{X,A}\otimes_{max} \POVM^{Y,B}$. Hence
\cref{def:co_abstract_self-test} subsumes \cref{def:abstract_self-test}.  While
\cref{def:abstract_self-test} is for finite-dimensional states on
$\POVM^{X,A}\otimes_{min} \POVM^{Y,B}$, \cref{def:co_abstract_self-test} also
gives us a definition of self-testing for non-finite-dimensional states on
$\POVM^{X,A}\otimes_{min} \POVM^{Y,B}$ (or in other words, for correlations $p
\in C_{qa}$ which do not have a finite-dimensional model). This is potentially
useful in understanding how self-testing interacts with limiting behaviour of
finite-dimensional models. For instance, Man\v{c}inska and Schmidt \cite{MS22} give
an example of a nonlocal game which is a non-robust self-test by combining a
finite-dimensional self-test with a game with a perfect $C_{qa}$ strategy, but
no perfect $C_q$ strategy. In our language, this nonlocal game is a
an abstract state self-test for finite-dimensional states on $\POVM^{X,A}\otimes_{min} \POVM^{Y,B}$,
but not a self-test for all states on $\POVM^{X,A}\otimes_{min} \POVM^{Y,B}$. 
Another example is the open question at the end of \Cref{sec:tensor_product_models},
which in this language asks whether every self-test for tensor-product models is
an abstract state self-test for states on $\POVM^{X,A}\otimes_{min} \POVM^{Y,B}$.

An immediate question about \cref{def:co_abstract_self-test} is whether there
is an equivalent description of this type of self-test in terms of models. For
this purpose, we make the following definitions.
\begin{definition}\label{def:submodel2}
Let 
\begin{align*}
    &\qcmodel \text{ and }\\
    &\wtdqcmodel
\end{align*}
be two commuting operator models.
\begin{enumerate}
    \item We say $S$ and $\wtd{S}$ are \textbf{equivalent}, and write $S\cong \wtd{S}$, if there exists a unitary $U:H\arr\wtd{H}$ such that
        \begin{enumerate}[(i)]
            \item $U\ket{\psi}=\ket{\wtd{\psi}}$, and 
            \item $UM^x_aU^*=\wtd{M}^x_a$ and $UN^y_b U^*=\wtd{N}^y_b$ for all $(a,b,x,y)\in A\times B\times X\times Y$.
        \end{enumerate}
    \item $\wtd{S}$ is a \textbf{submodel} of $S$ if there is a self-adjoint projection $\Pi
        \in \msB(H)$ commuting with $M^x_a$ and $N^y_b$ for all $(a,b,x,y)\in
        A\times B\times X\times Y$, such that $\wtd{H} = \Pi H$, $\ket{\wtd{\psi}} = \ket{\psi}$, and $\wtd{M}^x_a=\Pi M^x_a\Pi$ and
        $\wtd{N}^y_b=\Pi N^y_b\Pi$ for all $(a,b,x,y) \in A \times B \times X \times Y$. 

    \item $S$ is said to be \textbf{degenerate} if there exists a non-trivial
        self-adjoint projection $\Pi\in \msB(H)$ such that $\Pi\ket{\psi}=\ket{\psi}$ and
        $[\Pi,M_a^x]=[\Pi,N_b^y]=0$ for all $(a,b,x,y)\in A\times B\times X\times Y$.
        A commuting operator model is said to be \textbf{nondegenerate} if it
        is not degenerate. 
\end{enumerate}
A class of commuting operator models $\mcC$ is \textbf{closed under submodels} if
\begin{enumerate}[(i)]
    \item for any $S\in\mcC$, if $\wtd{S}$ is a commuting operator model such that $\wtd{S}\cong S$ then $\wtd{S}\in\mcC$, and
    \item for any $S\in\mcC$, if $\wtd{S}$ is a submodel of $S$ then $\wtd{S}\in\mcC$.
\end{enumerate}
\end{definition}
Note that in the definition of submodel in \Cref{def:submodel2}, the condition $\ket{\wtd{\psi}} = 
\ket{\psi}$ is equivalent to the condition that $\Pi \ket{\psi} = \ket{\wtd{\psi}}$, since both
imply that $\ket{\psi}$ is in the support of $\Pi$. This makes this notion of submodel a bit different
from the definition of submodel in \Cref{def:submodel1}, which only requires $\Pi_A \otimes \Pi_B \ket{\psi} = \lambda \ket{\wtd{\psi}}$ for some non-zero complex number $\lambda$.
As a result:
\begin{remark}\label{rmk:samestate}
    Suppose $S$ and $\wtd{S}$ are commuting operator models. If $\wtd{S}$ is equivalent to or a submodel of $S$, then $f_{\wtd{S}} = f_S$. 
\end{remark}

The following proposition shows that any commuting operator model $S$ has a nondegenerate submodel $\wtd{S}$ which is the GNS representation for the abstract state $f_S$. 

\begin{proposition}\label{prop:nd-cyclic} 
A commuting operator model is nondegenerate if and only if it is a cyclic model.
Furthermore, every commuting operator model has a nondegenerate submodel.
\end{proposition}

\begin{proof}
 We use $\msA$ to denote $\POVM^{X,A} \otimes_{max} \POVM^{Y,B}$. Let $p\in C_{qc}(X,Y,A,B)$, and let 
\begin{align*}
    \qcmodel
\end{align*}
be a commuting operator model for $p$ with associated representation $\pi$.

For the ``if" direction, suppose $\ket{\psi}$ is cyclic for $\pi$. If $\Pi$
is a projection in the commutant $\pi(\msA)'$ such that $\Pi \ket{\psi} = \ket{\psi}$,
then $\Pi \pi(\alpha) \ket{\psi} = \pi(\alpha) \Pi \ket{\psi} = \pi(\alpha) \ket{\psi}$,
so $\pi(\alpha) \ket{\psi}$ is contained in the image of $\Pi$. We conclude that
$\Pi = \Id$, so $S$ must be non-degenerate. 

For the ``only if" direction, suppose $\ket{\psi}$ is not cyclic for $\pi$. Let
$H_0:=\big(\pi(\msA)\ket{\psi}\big)^{-}$, and let $\Pi \in \msB(H)$ be the
self-adjoint projection onto $H_0$. Then $\Pi \in \pi(\msA)'$ is a non-trivial
projection such that $\Pi \ket{\psi} = \ket{\psi}$, and hence $S$ is degenerate. 
Furthermore, $(H_0, \{ \Pi M^x_a \Pi : a \in A, x \in X\}, \{\Pi N^y_b \Pi : b
\in B, y \in Y\}, \ket{\psi})$ is a non-degenerate submodel of $S$. 
\end{proof}

We can now show that \cref{def:co_abstract_self-test} is equivalent to having
a unique nondegenerate commuting-operator model (and this holds for any class
of commuting operator models closed under submodels).
\begin{theorem}\label{thm:concretecost}
    Let $\mcC$ be a class of commuting operator models that is closed under submodels, and let $\mcS:=\{f_S: S\in\mcC\}$ be the set of states on
    $\POVM^{X,A} \otimes_{max} \POVM^{Y,B}$ induced by $\mcC$. Then $p\in C_{qc}$
    is a self-test for $\mcS$ if and only if there is a commuting operator model
    \begin{equation*}
        \wtd{S} = (\wtd{H}, \{\wtd{M}^x_a : a \in A, x \in X\}, \{\wtd{N}^y_b:b \in B, y \in Y\}, \ket{\wtd{\psi}})
    \end{equation*}
    for $p$ in $\mcC$, 
    such that for every other commuting operator model
    \begin{equation*}
        S = (H, \{M^x_a : a \in A, x \in X\}, \{N^y_b:b \in B, y \in Y\}, \ket{\psi}),
    \end{equation*}
    for $p$ in $\mcC$, there is a submodel of $S$ which is equivalent to $\wtd{S}$. 
\end{theorem}
\begin{proof}
    The ``if'' direction follows directly from \Cref{rmk:samestate}. For the
    ``only if'' direction, suppose $p$ is an abstract state self-test for
    $\mcS$, with unique state $f \in \mcS$. By definition $f = f_{S'}$ for some
    $S' \in \mcC$. By \Cref{prop:nd-cyclic}, $S'$ has a non-degenerate submodel
    $\wtd{S}$. Suppose $S$ is any other model in $\mcC$ for $p$, and let $S''$
    be a non-degenerate submodel of $S$. Then $S''$ and $\wtd{S}$ are both
    cyclic models with $f_{S''} = f_S = f_{S'} = f_{\wtd{S}}$. Since GNS representations
    are unique up to unitary equivalence, $S''$ is equivalent to $\wtd{S}$.
\end{proof}
Note that the above theorem does not require $p$ to be an extreme point in
$C_{qc}$. 

The class of finite-dimensional commuting operator models is closed under submodels, and thus \Cref{thm:concretecost} applies to this class. 
Every finite-dimensional state (resp. projective finite-dimensional state) on
$\POVM^{X,A}\otimes_{max}\POVM^{Y,B}$ comes from a finite-dimensional state
(resp. projective finite-dimensional state) on
$\POVM^{X,A}\otimes_{min}\POVM^{Y,B}$, and thus $p$ is an abstract state
self-test for finite-dimensional states (resp. projective finite-dimensional
states) on $\POVM^{X,A}\otimes_{max}\POVM^{Y,B}$ if and only if $p$ is an
abstract state self-test for finite-dimensional states (resp. projective
finite-dimensional states) on $\POVM^{X,A}\otimes_{min}\POVM^{Y,B}$.
As a result, if $p \in C_q$ is an extreme point, then $p$ is a self-test for
(POVM) quantum models if and only if $p$ has a unique nondegenerate commuting
operator model (up to unitary equivalence). If, in addition, there exists a
projective full-rank quantum model for $p$, then $p$ is a self-test for
projective quantum models if and only if $p$ has a unique nondegenerate
commuting operator model. This gives a new criterion for $p\in C_q$ to be a
self-test for (finite-dimensional) quantum models.

It follows that if $p \in C_q$ is a commuting operator
self-test, then $p$ is a self-test for (finite-dimensional) quantum models.
Tsirelson showed that a wide family of correlations in $C_q$ are in fact
commuting operator self-tests \cite{Tsi87}. To state this result, let
$Cor(X,Y)$ be the set of matrices $c\in \R^{X\times Y}$ for which there is a
Euclidean space $V$ and vectors $\{\ket{u_x}\}_{x\in X}$, $\{\ket{v_y}\}_{y\in
Y}$ in $V$ of norm at most $1$, such that $c_{x,y}=\braket{u_x|v_y}$ for all
$x,y\in X\times Y$. If $p\in C_{qc}(X,Y,\Z_2,\Z_2)$ then the matrix $c$ defined
by \begin{equation*}
c_{x,y}=\sum_{a,b\in \Z_2}(-1)^{a+b}p(a,b|x,y)\end{equation*}
is in $Cor(X,Y)$, since if $\big(H,\{M^x_a:a\in \Z_2,x\in X \},\{N^y_b:b\in \Z_2,y\in Y\},\ket{\psi} \big)$ is a commuting operator model for $p$, then $c_{x,y}=\braket{\psi|(M_0^x-M_1^x)(N_0^y-N_1^y)|\psi}$, where $\|M_0^x-M_1^x\|\leq 1$ and $\|N_0^y-N_1^y\|\leq 1$. We refer to $c$ as the \textbf{XOR correlation} associated with $p$. Tsirelson shows that the linear map
\begin{equation*}
C_q(X,Y,\Z_2,\Z_2)\to Cor(X,Y)\ :\ p\mapsto c
\end{equation*} restricts to an isomorphism $C_q^{unbiased}(X,Y,\Z_2,\Z_2)\to Cor(X,Y)$, where
\begin{align*}C_q^{unbiased}(X,Y,\Z_2,\Z_2)=\Big\{p\in C_q(X,Y,\Z_2,\Z_2): \sum_{b\in \Z_2}p(0,b|x,y)=\sum_{b\in \Z_2}p(1,b|x,y)\text{ and }&\\ \sum_{a\in \Z_2}p(a,0|x,y)=\sum_{a\in \Z_2}p(a,1|x,y),\ x\in X,\ y\in Y&\Big\}.\end{align*}

\begin{theorem}[\cite{Tsi87}]
	If $p$ is an extreme point of $C_q^{unbiased}$ and the associated XOR correlation $c$ has even rank, then $p$ is a commuting operator self-test.
\end{theorem}

\begin{proof}
Since the isomorphism $C_q^{unbiased} \to Cor(X,Y)$ is linear, it preserves
extreme points.  If $p$ is an extreme point of $C_q^{unbiased}$, then the
associated XOR correlation $c$ is an extreme point in $Cor(X,Y)$. By
\cite[Theorem 3.2]{Tsi87}, if $c$ has even rank then all nondegenerate
commuting operator models for $p$ are unitarily equivalent. Hence, there is a
unique state on $\POVM^{X,\Z_2} \otimes_{max} \POVM^{Y,\Z_2}$ achieving $p$.
\end{proof}

\begin{example}
It is well known that the unique optimal correlation for the CHSH game is an extreme point of $C_q^{unbiased}$ with associated XOR correlation matrix
\begin{equation*}\begin{pmatrix} \frac{1}{\sqrt{2}} & \frac{1}{\sqrt{2}} \\ \frac{1}{\sqrt{2}} & \frac{-1}{\sqrt{2}}\end{pmatrix}.\end{equation*}
Since the associated XOR correlation has rank $2$, the optimal CHSH correlation is a commuting operator self-test.
\end{example}

A modern proof that the CHSH correlation is a commuting operator self-test can be found in \cite{Frei22}, where it is also shown that the Mermin-Peres magic square and magic pentagram games are commuting operator self-tests.

For more examples of commuting operator self-tests, let $\text{CHSH}_\alpha$, $\alpha\in[0,2)$ be the family of \emph{tilted} CHSH games introduced in \cite{AMP12}. It is well-known that there is a unique correlation $p_\alpha\in C_{q}(\Z_2,\Z_2,\Z_2,\Z_2)$ achieving the optimal winning probability of $\text{CHSH}_\alpha$, and that $p_\alpha$ is a self-test for projective quantum models \cite{BP15,CGS17}. This self-testing result is instrumental in the separations of the correlation sets $C_q\subsetneq C_{qs}\subsetneq C_{qa}$ in \cite{CS18,Col20,Bei21}. \cref{cor:binaryequiv} shows that $p_\alpha$ is a self-test for POVM quantum models as well.


\begin{proposition}
	If $p_\alpha\in C_{q}(\Z_2,\Z_2,\Z_2,\Z_2)$ is the unique optimal correlation for $\text{CHSH}_\alpha$, $\alpha\in [0,2)$, then $p_\alpha$ is a commuting operator self-test.
\end{proposition}

\begin{proof}
Fix $\alpha\in [0,2)$. Let $a_0:=m^0_0-m^0_1,a_1:=m^1_0-m^1_1,b_0:=n^0_0-n^0_1$ , and $b_1:=n^1_0-n^1_1$ in $\POVM^{\Z_2,\Z_2}\otimes_{max} \POVM^{\Z_2,\Z_2}$, and let $ \eta:=\alpha a_0+a_0b_0+a_0b_1+a_1b_0-a_1b_1$. A state $f$ on $\POVM^{\Z_2,\Z_2}\otimes_{max} \POVM^{\Z_2,\Z_2}$ achieves $p_\alpha$ if and only if $f(\eta)=\sqrt{8+2\alpha^2}=:\lambda$. Bamps and Pironio \cite{BP15} show that
\begin{equation*}
2\lambda(\lambda-\eta)=r_1^2+r_2^2=r_3^2+r_4^2
\end{equation*}
in $\PVM^{\Z_2,\Z_2}\otimes_{max} \PVM^{\Z_2,\Z_2}$, where
\begin{align*}
r_1&:=\lambda-\eta,\\
r_2&:=\alpha a_1-a_0b_0+a_0b_1-a_1b_0-a_1b_1,\\
r_3&:=\big(2a_0-\frac{\lambda}{2}(b_0+b_1)+\frac{\alpha}{2}(a_0b_0+a_0b_1-a_1b_0+a_1b_1)\big),\text{ and } \\
r_4&:=\big(2a_1-\frac{\lambda}{2}(b_0-b_1)+\frac{\alpha}{2}(a_0b_0-a_0b_1-a_1b_0-a_1b_1)\big).
\end{align*}
Moving to $\POVM^{\Z_2,\Z_2}\otimes_{max} \POVM^{\Z_2,\Z_2}$, we see that
\begin{align*}
    2\lambda(\lambda-\eta)
    =&\ r_1^2+r_2^2+\frac{1}{2}(s_1+s_2+s_3+s_4)+2(s_5+s_6)\\
    =&\ r_3^2+r_4^2+\frac{1}{2}(s_1+s_2+s_3+s_4)+2(2-a_0^2-a_1^2)(2-b_0^2-b_1^2)\\
    &+\frac{1}{2}\big(\lambda-\alpha(a_0-a_1) \big)^2(1-b_0^2)+\frac{1}{2}\big(\lambda-\alpha(a_0+a_1) \big)^2(1-b_1^2),
\end{align*}
where 
\begin{align*}
    s_1&:=(\alpha+2b_0)^2(1-a_0^2),\ s_2:=(\alpha+2b_1)^2(1-a_0^2),\\
    s_3&:=(\alpha-2b_0)^2(1-a_1^2),\ s_4:=(\alpha-2b_1)^2(1-a_1^2),\\
    s_5&:=(2+a_0a_1+a_1a_0)(1-b_0^2),\text{ and } s_6:=(2-a_0a_1-a_1a_0)(1-b_1^2).   
\end{align*}
Suppose $f$ is a state on $\POVM^{\Z_2,\Z_2}\otimes_{max} \POVM^{\Z_2,\Z_2}$ that achieves the optimal correlation $p_\alpha$. Then $f\left(2\lambda(\lambda-\eta)\right)=0$. Since $s_1,\ldots,s_6,\lambda-\alpha(a_0-a_1),\lambda-\alpha(a_0+a_1),1-b_0^2,1-b_1^2$, and $(2-a_0^2-a_1^2)(2-b_0^2-b_1^2)$ are positive in $\POVM^{\Z_2,\Z_2}\otimes_{max} \POVM^{\Z_2,\Z_2}$, and $\lambda-\alpha(a_0-a_1),\lambda-\alpha(a_0+a_1)$ commute with $1-b_0^2,1-b_1^2$, we obtain that $f(r_i^2)=f(s_j)=0$ for $1\leq i \leq 4$ and $1\leq j \leq 8$ where 
\begin{align*}
    s_7:=\big(\lambda-\alpha(a_0-a_1) \big)(1-b_0^2), \text{ and } s_8:=\big(\lambda-\alpha(a_0+a_1) \big)(1-b_1^2).
\end{align*}
Let $(\pi,H,\ket{\psi})$ be the GNS representation of $f$, and let $\mcL$ be the left ideal generated by $r_1,\ldots,r_4$ and $s_1,\ldots,s_8$. It follows that $\pi(l)\ket{\psi}=0$ for all $l\in\mcL$. Observe that
\begin{align*}
    &(4-\alpha^2)(a_0a_1+a_1a_0)
    =\big(b_0(r_3+r_4)-s_7+2\lambda\big)\big(b_0(r_3+r_4)-s_7\big)+s_1+s_3+4s_5,
\end{align*}
so $a_0a_1+a_1a_0\in\mcL$. This also implies
\begin{align*}
   1-b_0^2=\frac{1}{2}\big(s_5-(1-b_0^2)(a_0a_1+a_1a_0)\big)\in\mcL. 
\end{align*}
Likewise, one can verify that $1-a_0^2,1-a_1^2,$ and $1-b_1^2$ are in $\mcL$. Furthermore,
\begin{align*}
    \lambda a_0-2(b_0+b_1)=a_0r_1+a_1r_2-(b_0+b_1)\big((1-a_0^2)+(1-a_1^2)\big)+(b_0-b_1)(a_0a_1+a_1a_0)
\end{align*}
is in $\mcL$. Similarly, $\delta a_1-2(b_0-b_1)\in\mcL$ where $\delta:=\sqrt{8-2\alpha^2}$. Let $A_0:=\pi(a_0),A_1:=\pi(a_1),B_0:=\pi(b_0)$, and $B_1:=\pi(b_1)$. Then  \begin{align*}
(A_0A_1+A_1A_0)\ket{\psi}=&\ (\Id-A_0^2)\ket{\psi}=(\Id-A_1^2)\ket{\psi}=(\Id-B_0^2)\ket{\psi}=(\Id-B_1^2)\ket{\psi}\\=&\ \big(\lambda A_0-2(B_0+B_1)\big)\ket{\psi}= \big(\delta A_1-2(B_0-B_1)\big)\ket{\psi}=0.
\end{align*} 
For any monomial $a_{i_1}\cdots a_{i_k}b_{j_1}\cdots b_{j_l}$ in $\POVM^{\Z_2,\Z_2}\otimes_{max} \POVM^{\Z_2,\Z_2}$ we have
\begin{align}
    A_{i_1}\cdots A_{i_k}B_{j_1}\cdots B_{j_l}\ket{\psi} &=B_{j_1}\cdots B_{j_l}\frac{2(B_0+(-1)^{i_k}B_1)}{\sqrt{8+2(-1)^{i_k}\alpha^2}}\cdots\frac{2(B_0+(-1)^{i_1}B_1)}{\sqrt{8+2(-1)^{i_1}\alpha^2}}\ket{\psi}\label{eq:B}\\
    &=A_{i_1}\cdots A_{i_k}\frac{\lambda A_0+(-1)^{j_l}\delta A_1}{4}\cdots \frac{\lambda A_0+(-1)^{j_1}\delta A_1}{4}\ket{\psi}.\label{eq:A}
\end{align}
Hence by \cref{eq:B}, $(\Id-A_0^2)A_{i_1}\cdots A_{i_k}B_{j_1}\cdots B_{j_l}\ket{\psi}=0$, and
since $\ket{\psi}$ is cyclic for $\pi$, we conclude that $A_0^2=\Id$. Similarly, we have $A_1^2=B_0^2=B_1^2=\Id$ and $A_0A_1+A_1A_0=0$. Thus, the $C^*$-algebra $C^*\ang{A_0,A_1}$ generated by $A_0$ and $A_1$ is the Clifford algebra of rank $2$ which has a unique tracial state. Given $w_1:=A_{i_1}\cdots A_{i_k}$ and $w_{2}:=A_{j_1}\cdots A_{j_l}$,
\begin{align*}
    \bra{\psi}w_1w_2\ket{\psi}&=\bra{\psi}A_{i_1}\cdots A_{i_k}\frac{2(B_0+(-1)^{j_l}B_1)}{\sqrt{8+2(-1)^{j_l}\alpha^2}}\cdots\frac{2(B_0+(-1)^{j_1}B_1)}{\sqrt{8+2(-1)^{j_1}\alpha^2}} \ket{\psi}\\
    &= \bra{\psi}\frac{2(B_0+(-1)^{j_l}B_1)}{\sqrt{8+2(-1)^{j_l}\alpha^2}}\cdots\frac{2(B_0+(-1)^{j_1}B_1)}{\sqrt{8+2(-1)^{j_1}\alpha^2}}A_{i_1}\cdots A_{i_k}\ket{\psi}\\
    &= \bra{\psi}A_{j_1}\cdots A_{j_l}A_{i_1}\cdots A_{i_k}\ket{\psi}=\bra{\psi}w_2w_1\ket{\psi},
\end{align*}
so $w\mapsto\bra{\psi}w\ket{\psi}$ must be the unique tracial state on $C^*\ang{A_0,A_1}$.
By \cref{eq:A}, there is a unique state on $\POVM^{\Z_2,\Z_2}\otimes_{max} \POVM^{\Z_2,\Z_2} $ achieving $p_\alpha$, so $p_\alpha$ is a commuting operator self-test.
\end{proof}

 We leave it as an open problem to show that there are commuting operator self-tests which do not have finite-dimensional models.

\medskip

\newcommand{\etalchar}[1]{$^{#1}$}

\end{document}